\documentclass[pdftex,12pt,a4paper]{article}
\usepackage[pdftex]{graphicx}
\usepackage{curves,epic,latexsym,amsmath,amssymb,array,cite,ifthen}
\usepackage{epsfig}

\newcommand{\eqdef}{\stackrel{\scriptscriptstyle\bigtriangleup}{=} }

\newcommand{\Z}{\mathbb{Z}}

\newcommand{\argmin}{\operatornamewithlimits{argmin}}

\newcommand{\wH}{\mathrm{w_H}}
\newcommand{\dH}{\mathrm{d_H}}
\newcommand{\tH}{\mathrm{t_H}}
\newcommand{\dminH}{\mathrm{d_{minH}}}
\newcommand{\wD}{\mathrm{w_D}}
\newcommand{\dD}{\mathrm{d_D}}
\newcommand{\tD}{\mathrm{t_D}}
\newcommand{\dminD}{\mathrm{d_{minD}}}

\newcommand{\wminD}{\mathrm{w_{minD}}}

\newcounter{proglinecounter}
\newenvironment{pseudocode}%
    {\setcounter{proglinecounter}{0}%
     \begin{tabbing}12345\=12345\=123\=123\=123\=123\=123\=123\=123\=123\=123\= \kill}%
    {\end{tabbing}}
\newcommand{\npcl}[1][]
    {\>\refstepcounter{proglinecounter}\arabic{proglinecounter}%
     \ifthenelse{\equal{#1}{}}{}{\label{#1}}\' \>}
\newcommand{\pkw}[1]{\textbf{#1}}    

\newcounter{examplecntr}
\newenvironment{example}[1][]%
{\begin{trivlist}\item[]\refstepcounter{examplecntr}%
 {\bfseries Example~\theexamplecntr%
  \ifthenelse{\equal{#1}{}}{}{ (#1)}.
}}%
{\hfill$\Box$\end{trivlist}}

\newcounter{definitioncntr}
\newenvironment{definition}%
{\begin{trivlist}\item[]\refstepcounter{definitioncntr}%
{\bfseries Definition~\thedefinitioncntr.}}%
{\hfill$\Box$\end{trivlist}}

\newcounter{theoremcntr}
\newenvironment{theorem}[1][]%
{\begin{trivlist}\item[]\refstepcounter{theoremcntr}%
{\bfseries Theorem~\thetheoremcntr%
  \ifthenelse{\equal{#1}{}}{}{ (#1)}.
}}%
{\hfill$\Box$\end{trivlist}}

\newcounter{lemmacntr}
\newenvironment{lemma}[1][]%
{\begin{trivlist}\item[]\refstepcounter{lemmacntr}%
{\bfseries Lemma~\thelemmacntr%
  \ifthenelse{\equal{#1}{}}{}{ (#1)}.
}}%
{\hfill$\Box$\end{trivlist}}

\newcounter{corollarycntr}
{\begin{trivlist}\item[]\refstepcounter{corollarycntr}%
{\bfseries Corollary~\thecorollarycntr%
  \ifthenelse{\equal{#1}{}}{}{ (#1)}.
}}%
{\hfill$\Box$\end{trivlist}}

\newenvironment{proof}{\begin{trivlist}\item[]{\bfseries Proof: }
 }{\hfill$\Box$\end{trivlist}}

\newenvironment{proofof}[1]{\begin{trivlist}\item[]{\bfseries Proof #1: }
 }{\hfill$\Box$\end{trivlist}}

\newcommand{\eproofnegspace}{\\[-1.5\baselineskip]\rule{0em}{0ex}}

\newcounter{propositioncntr}
{\begin{trivlist}\item[]\refstepcounter{propositioncntr}%
{\bfseries Proposition~\thepropositioncntr%
  \ifthenelse{\equal{#1}{}}{}{ (#1)}.
}}%
{\hfill$\Box$\end{trivlist}}

\begin{document}
\hfill\texttt{January~5, 2012}

\begin{center}
\LARGE\bf On Polynomial Remainder Codes \vspace{2ex}
\end{center}

\begin{center}
Jiun-Hung~Yu and Hans-Andrea~Loeliger \\
Dept. of Information Technology and Electrical Engineering\\
ETH Zurich, Switzerland\\
Email: \{yu, loeliger\}@isi.ee.ethz.ch
\vspace{2ex}
\end{center}

{
\renewcommand{\thefootnote}{}
\footnotetext{A preliminary version of this work was presented in part in \cite{YuLoeliger}.}
}

\begin{center}
{\bf Abstract}
\end{center}

Polynomial remainder codes are a large class of codes
derived from the Chinese remainder theorem
that includes Reed-Solomon codes as a special case.
In this paper, we revisit these codes and study them
more carefully than in previous work.
We explicitly allow the code symbols to be
polynomials of different degrees,
which leads to two different notions of weight and distance.

Algebraic decoding is studied in detail.
If the moduli are not irreducible, the notion of an error locator
polynomial is replaced by an error factor polynomial.
We then obtain a collection of gcd-based decoding algorithms,
some of which are not quite standard even when specialized to Reed-Solomon codes.
\vspace{2ex}


\section{Introduction}

Polynomial remainder codes are a large class of codes
derived from the Chinese remainder theorem.
Such codes
were proposed by Stone \cite{Stone}, who also
pointed out that these codes include Reed-Solomon codes
\cite{Reed} as a special case. Variations of Stone's codes
were studied in \cite{Bossen,Mandelbaum,Mandelbaum2}. In
\cite{Stone} and \cite{Bossen}, the focus is on codes with a fixed
symbol size, i.e., the moduli are relatively prime polynomials of
the same degree.
A generalization of such codes was proposed by Mandelbaum \cite{Mandelbaum},
who also pointed out that using moduli of different
degrees can be advantageous for burst error correction
\cite{Mandelbaum2}.

Although the codes in \cite{Stone,Bossen,
Mandelbaum, Mandelbaum2} can, in principle, correct many random
errors, no efficient decoding algorithm for random errors was
proposed in these papers. In 1988, Shiozaki \cite{Shiozaki}
proposed an efficient decoding algorithm for Stone's codes
\cite{Stone} using Euclid's algorithm, and he also adapted this
algorithm to decode Reed-Solomon codes. However, the algorithm of
\cite{Shiozaki} is restricted to codes with a fixed symbol size,
i.e., fixed-degree moduli. Moreover, the argument given in
\cite{Shiozaki} seems to assume that all the moduli are
irreducible although this assumption is not stated explicitly.

In \cite{Mandelbaum3}, Mandelbaum made the interesting observation
that polynomial remainder codes (generalized as in \cite{Mandelbaum})
contain Goppa codes \cite{Goppa} as a special case.
By means of this observation,
generalized versions of Goppa codes
such as in \cite{BeSh:ggcISIT997c}
may also be viewed as polynomial remainder codes.
In subsequent work \cite{Mandelbaum4,Mandelbaum5},
Mandelbaum actually used the term ``generalized Goppa codes''
for (generalized) polynomial remainder codes.
He also proposed a decoding algorithm for such codes
using a continued-fractions approach \cite{Mandelbaum4,Mandelbaum5}.
However, this connection between (generalized) polynomial remainder codes and Goppa codes
will not be further pursued in this paper.

There is also a body of work on Chinese remainder codes over
integers, cf.\ \cite{Goldreich,Guruswami2}. However, the results
of the present paper are not directly related to that work.

In this paper, we revisit polynomial remainder codes as in
\cite{Stone}. We explicitly allow moduli of different degrees
(i.e., variable symbol sizes) within a codeword. In this way, we
can, e.g., lengthen a Reed-Solomon code by adding some
higher-degree symbols without increasing the size of the
underlying field. In consequence, we obtain two different notions
of distance---Hamming distance and degree-weighted distance---and
the corresponding minimum-distance decoding rules. Algebraic
decoding as in \cite{Shiozaki} is studied in detail. If the moduli
are not irreducible, the notion of an error locator polynomial is
replaced by an error factor polynomial. We then obtain a
collection of gcd-based decoding algorithms, some of which are not
quite standard even when specialized to Reed-Solomon codes.

This paper is organized as follows. In
Section~\ref{section:Chinese Remainder Codes}, we recall the
Chinese remainder theorem and the definition of Chinese remainder
codes over integers and polynomials. We also discuss erasures-only
decoding, i.e., the recovery of a codeword from a subset of its
symbols, for which we propose a method that appears to be new. In
Section~\ref{section:Polynomial Remainder Codes}, we focus on
polynomial remainder codes and their minimum-distance decoding,
both for Hamming distance and degree-weighted distance.
In Section~\ref{section: Locator Factor Key Equation}, we
introduce error locator polynomials and error factor polynomials
and a key equation for the latter.
In Section~\ref{section:exgcd}, we derive gcd-based decoding algorithms.
A synopsis of these algorithms is given in Section~\ref{section:DecodingSummary},
and their relation to prior work is discussed in Section~\ref{sec:PriorGCDDecoding}.
Section~\ref{section:Conclusion} concludes the paper.

The cardinality of a set $S$ will be denoted by $|S|$
and the absolute value of an integer $n$ will be denoted by $|n|$.
In Section~\ref{sec:ErasuresDecoding},
this same symbol
will also be used for the degree of a polynomial,
i.e., $|a(x)| \eqdef \deg a(x)$.

\section{Chinese Remainder Codes}
\label{section:Chinese Remainder Codes}

\subsection{Chinese Remainder Theorem and Codes}

Let $R=\Z$ or $R=F[x]$ for some field $F$.
(Later on, we will focus on $R=F[x]$.)
For $R=\Z$, for any positive $m\in\Z$,
let $R_m$ denote the ring $\{0,1,2,\ldots,m-1\}$ with addition and multiplication modulo~$m$;
for $R=F[x]$, for any monic polynomial $m(x)\in F[x]$,
let $R_m$ denote the ring of polynomials over $F$
of degree less than $\deg m(x)$ with addition and
multiplication modulo $m(x)$.
For $R=\Z$, $\gcd(a,b)$ denotes the greatest common divisor of $a,b\in\Z$,
not both zero;
for $R=F[x]$, $\gcd(a,b)$ denotes the monic polynomial of largest degree
that divides both $a,b\in F[x]$,
not both zero.

We will need the Chinese remainder theorem \cite{Stone} in the
following form.

\begin{theorem}[Chinese Remainder Theorem] \label{theorem:CRT}
\sloppy For some integer $n>1$, let $m_0,m_1,\ldots,m_{n-1} \in R$
be relatively prime (i.e., gcd$(m_i,m_j)=1$ for $i \neq j$) and
let $M_n \eqdef \prod_{i=0}^{n-1}m_i$. Then the mapping
\begin{eqnarray}
      \psi : R_{M_n} \rightarrow R_{m_0}\times \ldots \times
      R_{m_{n-1}}: a \mapsto
       \psi(a) \eqdef \big( \psi_0(a),\ldots,\psi_{n-1}(a) \big)  \label{eqn:defPsi}
\end{eqnarray}
with $\psi_i(a) \eqdef a \bmod m_i$ is a ring isomorphism.

The inverse of the mapping (\ref{eqn:defPsi}) is
\begin{eqnarray}
      \psi^{-1} : R_{m_0}\times \ldots \times R_{m_{n-1}} \rightarrow R_{M_n}:
         (c_0,\ldots,c_{n-1}) \mapsto
      \sum_{i=0}^{n-1} c_i \beta_i  \bmod M_n \label{eqn:c.2}
\end{eqnarray}
with coefficients
\begin{equation}
     \beta_i= \frac{M_n}{m_i} \cdot \left(\frac{M_n}{m_i}\right)_{\bmod
     m_i}^{-1}  \label{eqn:betai}
\end{equation}
where $(b)_{\bmod m_i}^{-1}$ denotes the inverse of $b$ in
$R_{m_i}$.
\end{theorem}

\begin{definition}
\label{def:CRT_Code}
A \emph{Chinese remainder code (CRT Code)} over $R$ is a set of the form
\begin{equation}
      C \eqdef \{(c_0,\ldots,c_{n-1}):
      c_i=a \bmod m_i~\text{for some}~a \in R_{M_k}\}
      \label{eqn:CRTC}
\end{equation}
where $n$ and $k$ are integers satisfying $1\leq k \leq n$,
where $m_0, m_1, \ldots, m_{n-1} \in R$ are relatively prime,
and where $M_k \eqdef \prod_{i=0}^{k-1}m_i$.
\end{definition}
In other words, a CRT code consists of the images $\psi(a)$, with
$\psi$ as in (\ref{eqn:defPsi}), of all $a \in R_{M_k}$. For
$R=F[x]$, CRT codes are linear (i.e., vector spaces) over $F$; for
$R=\Z$, however, CRT codes are not linear since the pre-image of
the sum of two codewords may exceed the range of $M_k$.

The components $c_{i}=\psi_i(a)$ in (\ref{eqn:defPsi}) and
(\ref{eqn:CRTC}) will be called \emph{symbols}. Note that each
symbol is from a different ring $R_{m_i}$; these rings need not
have the same number of elements. We will often (but not always)
assume that the moduli $m_i$ in Definition \ref{def:CRT_Code}
satisfy the condition
\begin{equation}
|R_{m_0}| \leq |R_{m_1}| \leq \ldots \leq |R_{m_{n-1}}|. \label{eqn:alphabetCondition}
\end{equation}
We will refer to (\ref{eqn:alphabetCondition}) as the \emph{Ordered-Symbol-Size Condition.}

\subsection{Interpolation}
\label{sec:ErasuresDecoding}

Consider the problem of reconstructing a codeword
$c=(c_0,\ldots,c_{n-1})$ from a subset of its symbols.
Specifically, let $C$ be a CRT code as in
Definition~\ref{def:CRT_Code} and let $S$ be a subset of $\{
0,1,2, \ldots, \mbox{$n-1$}\}$ with cardinality $|S| >0$. Let
$c=(c_0,\ldots,c_{n-1}) = \psi(a) \in C$ be the codeword
corresponding to some $a\in R_{M_k}$ by~(\ref{eqn:CRTC}). Suppose
we are given $\tilde{c} = (\tilde{c}_0, \ldots \tilde{c}_{n-1})$
with
\begin{equation}
\tilde{c}_i = c_i ~~\text{for $i\in S$}
\end{equation}
(and with arbitrary $\tilde{c}_i\in R_{m_i}$ for $i\not\in S$) and
we wish to reconstruct $a=\psi^{-1}(c)$ from $\tilde{c}$.
This problem arises, for example, when the channel
erases some symbols (and lets the receiver know the erased positions)
but delivers the other symbols unchanged.
However, this problem also arises as the last step
in the decoding procedures that will be discussed later in the paper.

This interpolation problem can certainly be solved if $S$ is
sufficiently large. A first solution follows immediately from the
CRT (Theorem~\ref{theorem:CRT}). Specifically, with $M_S \eqdef
\prod_{i\in S} m_i$, Theorem~\ref{theorem:CRT} can be applied as
follows: if
\begin{equation}
|M_S| \geq |M_k|  \label{inq:recons.cond.0}
\end{equation}
then
\begin{equation}
   a = \sum_{i=0}^{n-1} \tilde{c}_i \tilde{\beta}_i \bmod M_S \label{eqn:recon.1}
\end{equation}
with
\begin{equation}  \label{eqn:recon.2}
 \tilde{\beta}_i \eqdef \left\{ \begin{array}{ll}
    \frac{M_S}{m_i} \cdot
      \left(\frac{M_S}{m_i}\right)_{\bmod m_i}^{-1}, & i\in S \\
      0, &  i\not\in S.
 \end{array} \right.
\end{equation}

Obviously, the coefficients $\tilde\beta_i$ in (\ref{eqn:recon.2})
depend on the support set $S$. Interestingly, there is a second
solution to the interpolation problem that avoids the computation
of these coefficients: the following theorem shows how
$a=\psi^{-1}(c)$ can be computed from $\psi^{-1}(\tilde{c})$,
which in turn may be computed using the fixed
coefficients~(\ref{eqn:betai}).

\begin{theorem}[Fixed-Transform Interpolation] \label{theorem:reconstruction}
If
\begin{equation} \label{inq:recons.cond.1}
|M_S| \geq |M_k|
\end{equation}
then
\begin{equation}
      \psi^{-1}(c) = Z/M_{\overline{S}}  \label{eqn:recon.4}
\end{equation}
where $M_{\overline{S}} \eqdef M_n/M_S$ and where
\begin{equation}
  Z \eqdef (M_{\overline{S}} \cdot \psi^{-1}(\tilde{c})) \bmod M_n \label{eqn:recon.3}
\end{equation}
is a multiple of $M_{\overline{S}}$.
\end{theorem}
This theorem does not appear in standard expositions of the CRT;
perhaps it is new.
Its application to coding, even to Reed-Solomon codes
(cf.\ Section~\ref{sec:ErasuresDecodingPRC}),
also appears to be new.

\begin{proofof}{of Theorem~\ref{theorem:reconstruction}}
Let $\bar{c} \eqdef c-\tilde{c}$, let $\bar{a} \eqdef
\psi^{-1}(\bar{c})$, and note that $\psi^{-1}(\tilde{c}) = (a -
\bar{a}) \bmod M_n$. Note also that $|M_{\overline{S}} \cdot a| <
|M_n|$ because of (\ref{inq:recons.cond.1}). Then
\begin{align}
      Z &= \left( M_{\overline{S}} \cdot (a-\bar{a})  \right) \bmod M_n \\
        &= M_{\overline{S}}\cdot a - (M_{\overline{S}}\cdot \bar{a})  \bmod M_n  \\
        &= M_{\overline{S}}\cdot a \label{eqn:pfrecon.1}
\end{align}
where the last step follows from
\begin{align}
\psi(M_{\overline{S}}\cdot \bar{a})
& =  \psi(M_{\overline{S}}) \psi(\bar{a}) \\
& =  0.
\end{align}
\eproofnegspace
\end{proofof}

\subsection{Hamming Distance and Singleton Bound}

For any $a \in R_{M_n}$, the Hamming weight of $\psi(a)$ (i.e.,
the number of nonzero symbols $\psi_i(a)$, \mbox{$0 \leq i \leq
n-1$}) will be denoted by $\wH(\psi(a))$. For any $a,b \in
R_{M_n}$, the Hamming distance between $\psi(a)$ and $\psi(b)$
will be denoted by $\dH(\psi(a),\psi(b)) \eqdef \wH(\psi(a) -
\psi(b))$.
The minimum Hamming distance of a CRT code $C$ will be denoted
by~$\dminH(C)$.
\begin{theorem} \label{theorem:Hamming-weight}
Let $C$ be a CRT code as in Definition~\ref{def:CRT_Code}
satisfying the Ordered-Symbol-Size
Condition~(\ref{eqn:alphabetCondition}). Then the Hamming weight
of any nonzero codeword $\psi(a)$ ($a\in R_{M_k}$, $a\neq 0$)
satisfies
\begin{equation} \label{eqn:MinWeight}
\wH(\psi(a)) \geq n-k+1
\end{equation}
and
\begin{equation} \label{eqn:MinDist}
\dminH(C) = n-k+1.
\end{equation}
\eproofnegspace
\end{theorem}

\begin{proof}
For any nonzero $a\in R_{M_n}$, assume that the image $\psi(a)$
has Hamming weight $\wH(\psi(a)) \leq n-k$,
i.e., the number of zero symbols of $\psi(a)$ is at least $k$. For $R=\Z$,
this implies $a \geq M_k$;
for $R=F[x]$, this implies $\deg a\geq \deg M_k$.
In both cases, $a\not\in R_{M_k}$, which proves~(\ref{eqn:MinWeight}).

As for (\ref{eqn:MinDist}),
consider $\dH(\psi(a),\psi(b))$
for any $a,b\in R_{M_k}$, $a\neq b$.
For $R=F[x]$, $a-b\in R_{M_k}$ and thus
\begin{align}
\dH(\psi(a),\psi(b))
 & = \wH(\psi(a)-\psi(b)) \\
 & = \wH(\psi(a-b)) \\
 & \geq n-k+1
\end{align}
by (\ref{eqn:MinWeight}). For $R=\Z$, either $a-b\in R_{M_k}$ or
$b-a\in R_{M_k}$ and the same argument applies. It follows that
$\dminH(C) \geq n-k+1$. Finally, the equality in
(\ref{eqn:MinDist}) follows from the Singleton bound below.
\end{proof}

In the following theorem, we will use the following notation.
For any subset $S\subset \{ 0, 1, \ldots, n-1 \}$,
let $\overline{S} \eqdef \{ 0, 1, \ldots, n-1 \} \setminus S$
and let
\begin{equation} \label{eqn:RS}
R_S \eqdef \bigotimes_{i\in S} R_{m_i},
\end{equation}
the direct product of all rings $R_{m_i}$ with $i\in S$.

\begin{theorem}[Singleton Bound for Hamming Distance]\label{theorem:SingletonBounddminH}
Let $C$ be a code in $R_{\{0,\ldots, n-1\}}$ (i.e., a nonempty
subset of $R_{m_0}\times \cdots \times R_{m_{n-1}}$) with minimum
Hamming distance $\dminH$.
Then
\begin{equation} \label{eqn:SingletonBounddminH}
|C| \leq  \min_{S \subset \{0,1,\ldots, n-1\}}\{|R_S|: |S|>n-\dminH  \}.
\end{equation}
\eproofnegspace
\end{theorem}
Note that this theorem does not require the Ordered-Symbol-Size Condition (\ref{eqn:alphabetCondition}).

\begin{proof}
Let $\overline{S}$ be a subset of $\{0,1,\ldots, n-1\}$ with
$|\overline{S}| < \dminH$. For every word $c\in C$, erase its
components in $\overline{S}$. The resulting set of shortened
words, which are elements of $R_S$, has still $|C|$ elements.
\end{proof}

For CRT codes satisfying the Ordered-Symbol-Size
Condition~(\ref{eqn:alphabetCondition}),
we have $|C| = |R_{M_k}|$;
on the other hand, the right-hand side of (\ref{eqn:SingletonBounddminH}) becomes
\begin{equation}
|R_{ \{0,\ldots,n-\dminH\} }| = |R_{M_{n-\dminH+1}}|
\end{equation}
where $M_{n-\dminH+1}\eqdef \prod_{i=0}^{n-\dminH} m_i $. It then
follows from (\ref{eqn:SingletonBounddminH}) that $|R_{M_k}| \leq
|R_{M_{n-\dminH+1}}|$ and thus
\begin{equation}
k \leq n-\dminH+1.
\end{equation}

\section{Polynomial Remainder Codes}
\label{section:Polynomial Remainder Codes}

From now on, we will focus on the case $R=F[x]$ for some finite
field $F$.

\subsection{Definition and Some Examples}

\begin{definition}
\label{def:PR_Code} A \emph{polynomial remainder code} is a CRT
code over $R=F[x]$ with monic moduli $m_i(x)$, i.e., a set of the
form
\begin{equation}
      C =\{(c_0,\ldots,c_{n-1}):
      c_{i}=a(x) \bmod m_i(x)~\text{for some}~a(x) \in R_{M_k}\}. \label{eqn:PRC}
\end{equation}
A polynomial remainder code is \emph{irreducible} if the
polynomials $m_0(x),\ldots,m_{n-1}(x)$ are all irreducible
\cite{YuLoeliger}.
\end{definition}
For such codes, the Ordered-Symbol-Size Condition~(\ref{eqn:alphabetCondition})
may be written as
\begin{equation} \label{eqn:OrderedDegreeCondition}
\deg m_0(x)\leq \deg m_1(x) \leq \ldots \leq \deg m_{n-1}(x),
\end{equation}
which we will call the \emph{Ordered-Degree Condition.}

\begin{example}[Binary Irreducible Polynomial Remainder Codes]
\label{ex:BinIrrPRC} Let $F=\text{GF}(2)$ be the finite field with
two elements and let $m_0(x),\ldots,m_{n-1}(x)$ be different
irreducible binary polynomials.

The number of irreducible binary polynomials of degree up to 16 is
given in Appendix~A. For example, by using only irreducible moduli
of degree 16, we can obtain a code with $\deg M_n(x) = 4080$; by
using irreducible moduli of degree up to 16, we can achieve $\deg
M_n(x) = \text{130'486}$.
\end{example}

\begin{example}[Polynomial Evaluation Codes and Reed-Solomon Codes]
\label{example:RS} Let $\beta_0,\beta_1,\ldots,\beta_{n-1}$ be
distinct elements of some finite field $F$ (which implies $n \leq
|F|$). A~\emph{polynomial evaluation code} over $F$ is a code of
the form
\begin{equation}
  C \eqdef \{(c_0,\ldots,c_{n-1}):
      c_{i}=a(\beta_i) ~\text{for some $a(x) \in F[x]$ of $\deg a(x)<k$} \}. \label{eqn:PEC}
\end{equation}
A Reed-Solomon code is a polynomial evaluation code with
$\beta_i=\alpha^i$, where $\alpha$ is a primitive $n$-th root of
unity in $F$. With
\begin{equation}
m_i(x) \eqdef x-\beta_i,
\end{equation}
a polynomial evaluation code may be viewed as a polynomial remainder code
since
\begin{equation}
      c_{i}=a(\beta_i)=a(x) \bmod m_i(x).  \label{RS.1}
\end{equation}
For Reed-Solomon codes (as defined above), we then have
\begin{equation}
M_n(x)=x^n-1.  \label{eqn:RS.2}
\end{equation}
\eproofnegspace
\end{example}

\begin{example}[Polynomial Extensions of Reed-Solomon Codes]
\label{example:PolyExtRS}
When Reed-Solomon codes are viewed as polynomial remainder codes
as in Example~\ref{example:RS}, the code symbols are constants, i.e.,
polynomials of degree at most zero.
Reed-Solomon codes can be extended with additional symbols in $F[x]$
by adding some moduli $m_i(x)$ of degree two (or higher).
\end{example}

\subsection{Degree-weighted Distance}

Let
\begin{equation} \label{eqn:defN}
N \eqdef \deg M_n(x) = \sum_{i=0}^{n-1} \deg m_i(x)
\end{equation}
and
\begin{equation} \label{eqn:defK}
K \eqdef \deg M_k(x) = \sum_{i=0}^{k-1} \deg m_i(x).
\end{equation}
Note that $K$ is the dimension of the code as a subspace of $F^N$.

\begin{definition} \label{def:DegreeWeight}
The \emph{degree weight of a set} $S \subset \{0,1,\ldots,n-1\}$ is
\begin{equation}
\wD(S) \eqdef \sum_{i \in S} \deg m_i(x) \label{eqn:wDS}.
\end{equation}
For any $a(x) \in R_{M_n}$, the \emph{degree weight} of $\psi(a)
=\big(\psi_0(a), \ldots,\psi_{n-1}(a)\big)$ is
\begin{equation}
 \wD(\psi(a)) \eqdef \sum_{i:\psi_i(a)\neq 0}\deg m_i,
\end{equation}
and for any $a(x),b(x) \in R_{M_n}$, the \emph{degree-weighted
distance} between $\psi(a)$ and $\psi(b)$ is
\begin{equation} \label{eqn:defdD}
 \dD(\psi(a),\psi(b)) \eqdef \wD(\psi(a)-\psi(b)).
\end{equation}
\eproofnegspace
\end{definition}
Note that the degree-weighted distance
satisfies the triangle inequality:
\begin{equation}
\dD(\psi(a),\psi(b))\leq
\dD(\psi(a),\psi(c)) + \dD(\psi(b),\psi(c))
\label{eqn:triangle-inequality}
\end{equation} for all $a(x),b(x),c(x)\in R_{M_n}$.

Let $\dminD(C)$
denote the \emph{minimum degree-weighted distance} of a polynomial
remainder code $C$, i.e.,
\begin{equation} \label{eqn:defdminD}
\dminD(C) \eqdef \min_{c,c' \in C:\, c \neq c'} \dD(c,c'),
\end{equation}
and let
\begin{equation}
\wminD(C) \eqdef \min_{c \in C:\, c \neq 0} \wD(c)
\end{equation}
be the minimum degree weight of any nonzero codeword.
We then have the following analog of
Theorem~\ref{theorem:Hamming-weight}:

\begin{theorem}[Minimum Degree-Weighted Distance] \label{theorem:Degree-weight}
Let $C$ be a code as in Def\-i\-ni\-tion~\ref{def:PR_Code}.
Then
\begin{align}
\dminD(C) & = \wminD(C)  \label{eqn:dminDeqwminD} \\
          & = \min_{S \subset \{0,\ldots, n-1\}} \big\{ \wD(S): \wD(S) > N-K  \big\}
           \label{eqn:MinDegreeWeight} \\
 & > N-K.    \label{eqn:MinDegreeWeightNmK}
\end{align}
\eproofnegspace
\end{theorem}

\noindent
If all moduli $m_i(x)$ have degree one, then the right-hand side of (\ref{eqn:MinDegreeWeight})
equals \mbox{$N-K+1$}.
Note also that unlike
Theorem~\ref{theorem:Hamming-weight},
Theorem~\ref{theorem:Degree-weight} does not require the
Ordered-Degree Condition~(\ref{eqn:OrderedDegreeCondition}).

\begin{proof}
Equation (\ref{eqn:dminDeqwminD}) is obvious from the linearity of
the code over $F$, and (\ref{eqn:MinDegreeWeightNmK}) is obvious
as well. It remains to prove (\ref{eqn:MinDegreeWeight}).

Let $d$ be the right-hand side of (\ref{eqn:MinDegreeWeight}). For
any nonzero $a(x)\in R_{M_k}$, assume that the image $\psi(a)$ has
degree weight $\wD(\psi(a)) \leq N-K$, i.e., the sum of $\deg
m_i(x)$ over the zero symbols of $\psi(a)$ is at least~$K$. Then
$\deg a(x)\geq K = \deg M_k(x)$, which is impossible since
$a(x)\in R_{M_k}$. We thus have $\wD(\psi(a)) > N-K$. It then
follows from Definition~\ref{def:DegreeWeight} that $\wD(\psi(a))
\geq d$ and thus $\wminD(C) \geq d$.

Conversely, let
$S$ be a subset of $\{0,1,\ldots,n-1\}$ such that
$\wD(S)=d$.
Then there
exists some nonzero $a(x)\in R_{M_k}$ such that $\psi_i(a)\neq 0$ for
each $i \in S$ but $\psi_j(a)=0$ for each
$j\in \{0,1,\ldots,n-1\}\setminus S$.
Thus $\wD(\psi(a)) = \wD(S) = d$, which implies $\wminD(C) \leq d$.
\end{proof}

\begin{theorem}[Singleton Bound for Degree-weighted Distance]
 Let $C$ be a nonempty subset of $R_{m_0}\times \cdots \times R_{m_{n-1}}$
 with minimum degree-weighted distance $\dminD$ and with $N$ as in (\ref{eqn:defN}). Then
 \begin{equation} \label{eqn:SingletonBounddminD}
 \log_F |C| \leq  \min_{S \subset \{0,\ldots, n-1\}}\{\wD(S): \wD(S)> N- \dminD \}.
 \end{equation}
\eproofnegspace
\end{theorem}

\begin{proof}
Recall the notation $\overline{S}$ and $R_S$ as in (\ref{eqn:RS}).
Let $\overline{S}$ be a subset of $\{0,1,\ldots, n-1\}$ with
$\wD(\overline{S}) < \dminD$. For every word $c\in C$, erase its
components in $\overline{S}$. The resulting set of shortened
words, which are elements of $R_S$, has still $|C|$ elements. Thus
$|C| \leq |R_S| = |F|^{\wD(S)}$, and
(\ref{eqn:SingletonBounddminD}) follows.
\end{proof}

For polynomial remainder codes, we have $\log_F|C| = K$
and (\ref{eqn:SingletonBounddminD}) holds with equality.
To see this, we first write (\ref{eqn:SingletonBounddminD}) as
\begin{equation}  \label{eqn:SingletonBounddminD.1}
 K \leq \min_{S \subset \{0,\ldots, n-1\}}\{\wD(S): \wD(S)> N- \dminD \}.
\end{equation}
On the other hand, for $S = \{0,\ldots,k-1\}$, we have $\wD(S) =
K$, and using (\ref{eqn:MinDegreeWeightNmK}), we obtain
\begin{equation}  \label{eqn:SingletonBounddminD.2}
\min_{S \subset \{0,\ldots, n-1\}}\{\wD(S): \wD(S)> N- \dminD \} \leq K.
\end{equation}
We thus have equality in (\ref{eqn:SingletonBounddminD.1})
and (\ref{eqn:SingletonBounddminD.2}), and therefore also in (\ref{eqn:SingletonBounddminD}).

In the special case where all the moduli
$m_0(x),\ldots,m_{n-1}(x)$ have the same degree,
the two Singleton bounds (\ref{eqn:SingletonBounddminD}) and
(\ref{eqn:SingletonBounddminH})
are equivalent.

\subsection{Interpolation and  Erasures Decoding}
\label{sec:ErasuresDecodingPRC}

We now return to the subject of Section~\ref{sec:ErasuresDecoding}
and specialize it to polynomial remainder codes.
Let $C$ be a code as in Definition~\ref{def:PR_Code}.
Let $c = (c_0,\ldots,c_{n-1}) =\psi(a(x)) \in C$
be the codeword corresponding to some polynomial $a(x) \in R_{M_k}$.
Let $S$ be a set of positions $i\in \{ 0, \ldots, n-1\}$ where $c_i$ is known.
Let $\tilde{c} = (\tilde{c}_0, \ldots, \tilde{c}_{n-1})$ satisfy
$\tilde{c}_i = c_i$ for $i\in S$ with arbitrary $\tilde{c}_i \in R_{m_i}$
for $i\not\in S$.
Suppose we wish to reconstruct $a(x)$ from $\tilde{c}$ and $S$.

Let $\overline{S} \eqdef \{0,\ldots,n-1\} \setminus S$
be the indices of the unknown components of $c$
and
let $M_{\overline{S}}(x)=\prod_{i \in \overline{S}}
m_i(x)$ as in Section~\ref{sec:ErasuresDecoding}.
Recall that $\wD(\overline{S})$ denotes the degree weight of the unknown (erased) components of $c$.
Then Theorem \ref{theorem:reconstruction} can be restated as follows:
\begin{theorem}[Fixed-Transform Interpolation for Polynomial Remainder Codes]
\label{theorem:Erasure Correction Bound}
If
\begin{equation} \label{eqn:era_bound}
 \wD(\overline{S}) \leq N-K,
\end{equation} then
\begin{equation}
      a(x)= Z(x)/M_{\overline{S}}(x) \label{eqn:erasure.0}
\end{equation}
with
\begin{equation} \label{eqn:erasure.1}
  Z(x) \eqdef M_{\overline{S}}(x) \psi^{-1}(\tilde{c}) \bmod M_n(x).
\end{equation}
\eproofnegspace
\end{theorem}

The equivalence of (\ref{eqn:era_bound}) and (\ref{inq:recons.cond.1}) follows
from noting that the left-hand side of (\ref{inq:recons.cond.1})
is $|M_S| = N - \wD(\overline{S})$ and the right-hand side of (\ref{inq:recons.cond.1})
is $|M_k| = K$.

Since $\overline{S}$ contains the support set of $\tilde{c} - c$,
the polynomial $M_{\overline{S}}(x)$ is a multiple of an error locator polynomial
(as will be defined in Section~\ref{section: Locator Factor Key Equation}).

In contrast to most other statements in this paper,
Theorem~\ref{theorem:Erasure Correction Bound} appears to be new
even when specialized to Reed-Solomon codes (as in
Example~\ref{example:RS}), where $M_n(x) = x^n-1$ and the modulo
operation in (\ref{eqn:erasure.1}) is computationally trivial.

\subsection{Minimum-Distance Decoding}

\label{sec:MinDistDec}

Let $C$ be a code as in Definition~\ref{def:PR_Code}. The receiver
sees $y=c+e$, where $c\in C$ is the transmitted codeword and $e$
is an error pattern. A~\emph{minimum Hamming distance decoder} is
a decoder that produces
\begin{equation}
      \hat{c}=\argmin_{c \in C} \dH(c,y).
      \label{eqn:minHammingDecoder}
\end{equation}
A~\emph{minimum degree-weighted distance decoder} is a
decoder that produces
\begin{equation}
      \hat{c}=\argmin_{c \in C} \dD(c,y).
      \label{eqn:minDegreeDecoder}
\end{equation}
In general, the decoding rules (\ref{eqn:minHammingDecoder})
and (\ref{eqn:minDegreeDecoder}) produce different estimates $\hat{c}$
as will be illustrated by the examples below.

\begin{theorem}[Basic Error Correction Bounds] \label{theorem:BsaicBondType-I}
If $\dH(c,y) < \dminH(C)/2$, then the rule
(\ref{eqn:minHammingDecoder}) produces $\hat{c}=c$. If $\dD(c,y) <
\dminD(C)/2$, then the rule (\ref{eqn:minDegreeDecoder}) produces
$\hat{c}=c$.
\end{theorem}

\begin{proof}
The proof follows the standard pattern; we prove only the second part.
Assume $\hat{c}\neq c$, which implies
$\dD(\hat{c},y) \leq \dD(c,y)$. Using
the triangle inequality~(\ref{eqn:triangle-inequality}), we obtain
$\dminD(C) \leq \dD(\hat{c},c)\leq
\dD(\hat{c},y) + \dD(c,y)\leq 2\dD(c,y)$.
\end{proof}
The second part of Theorem~\ref{theorem:BsaicBondType-I} can also
be formulated as follows: if
\begin{equation} \label{Basicbound:dminD}
\wD(e) \leq \tD\eqdef \left\lfloor \frac{N-K}{2} \right\rfloor,
\end{equation}
then the rule (\ref{eqn:minDegreeDecoder}) produces $\hat{c}=c$.
If the Ordered-Degree Condition~(\ref{eqn:OrderedDegreeCondition})
is satisfied, then the first part of
Theorem~\ref{theorem:BsaicBondType-I} implies the following: if
\begin{equation} \label{Basicbound:dminH}
\wH(e) \leq \tH\eqdef \left\lfloor \frac{n-k}{2} \right\rfloor,
\end{equation}
then the rule (\ref{eqn:minHammingDecoder}) produces $\hat{c}=c$.

Depending on the degrees $\deg m_i(x)$, it is possible that the
condition $\wH(e)\leq \tH$ implies $\wD(e)\leq \tD$ (see
Example~\ref{example:decoderAB.2} below). In general, however,
none of the two decoding rules (\ref{eqn:minHammingDecoder})
and~(\ref{eqn:minDegreeDecoder}) is uniformly stronger than the
other.

\begin{example} \label{example:decoderAB.1}
Let $k=3$ and $n=5$, and let $\deg m_i(x)=i$ for $i=1,2,\ldots,5$.
We then have $\tH=1$, $K=6$, $N=15$, and $\tD=4$.
Consider the following two decoders:
Decoder~A corrects all errors with $\wH(e) \leq \tH$
and
Decoder~B corrects all errors with $\wD(e) \leq \tD$.
We then observe:
\begin{itemize}
\item
Decoder~A corrects all single symbol errors in any position.
\item
Decoder~B corrects all single symbol errors in the first 4 symbols (but not in position~5),
and it corrects two symbol errors in positions 1 and~2,
or in positions 1 and~3.
\eproofnegspace
\end{itemize}
\end{example}

\begin{example} \label{example:decoderAB.2}
Let $k=3$ and $n=5$, and let
$\deg m_1(x)=\deg m_2(x)=\deg m_3(x)=1$ and
$\deg m_4(x)=\deg m_5(x)=2$.
We then have $\tH=1$, $K=3$, $N=7$, and $\tD=2$.
Considering the same decoders as in
Example~\ref{example:decoderAB.1}, we observe:
\begin{itemize}
\item Decoder~A corrects all single symbol errors in any position.
\item Decoder~B also corrects all single symbol errors, and in
addition, it corrects any two symbol errors in the first 3
symbols. \eproofnegspace
\end{itemize}
\end{example}

\subsection{Summary of Code Parameters}

Let us summarize the key parameters of a polynomial remainder code
$C$ both in terms of Hamming distance and in terms of
degree-weighted distance. For the latter, the code parameters are
$(N,K,\dminD)$ with $N$, $K$, and $\dminD$ defined as in
(\ref{eqn:defN}), (\ref{eqn:defK}), (\ref{eqn:defdminD}) and with
$\dminD$ as in~(\ref{eqn:MinDegreeWeight}).
By the \emph{rate} of the code, we mean the
quantity
\begin{equation} \label{eqn:defRate}
      \frac{1}{N}\log_{|F|}|C|=\frac{K}{N}
\end{equation}
where $F$ is the underlying field.

With respect to Hamming distance, we have the parameters $(n,
k,\dminH)$ and the \emph{symbol rate} $k/n$. If the code $C$
satisfies the Ordered-Degree Condition
(\ref{eqn:OrderedDegreeCondition}), we have $\dminH = n-k+1$.

In the special case where all the moduli
$m_0(x),\ldots,m_{n-1}(x)$ have the same
degree, the two triples $(N,K,\dminD)$ and $(n,k,\dminH)$
are equal up to a scale factor and the rate (\ref{eqn:defRate})
equals the symbol rate $k/n$.

\section{Error Factor Polynomial}
\label{section: Locator Factor Key Equation}

Decoding Reed-Solomon codes can be reduced to solving a key
equation that involves an error locator polynomial \cite{Roth}.
We are going to propose
such an approach for polynomial remainder codes.
As it turns out, in general (i.e., beyond irreducible remainder codes),
we will need a slight generalization
of an error locator polynomial.

Let $C$ be a polynomial remainder code of the form
(\ref{eqn:PRC}). For the received $y=c+e$, where
$c=(c_0,\ldots,c_{n-1}) \in C$ is a transmitted codeword, and
where $e=(e_0,\ldots,e_{n-1})$ is an error pattern, let
$Y(x)=a(x)+E(x)$ denote the pre-image $\psi^{-1}(y)$ of $y$
with $\psi^{-1}$ as in (\ref{eqn:c.2}), where $a(x)=\psi^{-1}(c)$
is the transmitted-message polynomial, and where $E(x)$ denotes
the pre-image $\psi^{-1}(e)$ of the error $e$.

\subsection{Error Factor Polynomial, Key Equation, and Interpolation}
 \label{section:Error factors}

\begin{definition} \label{def:errorfactor.1}
An \emph{error factor polynomial}
is a nonzero polynomial $\Lambda(x) \in F[x]$ such that
\begin{equation} \label{eqn:efEx}
  \Lambda(x) E(x)\bmod M_n(x)= 0.
\end{equation}
\eproofnegspace
\end{definition}
Clearly, the polynomial
\begin{equation} \label{eqn:ef}
  \Lambda_f(x) \eqdef \frac{M_n(x)}{\gcd\big(E(x),M_n(x)\big)}
\end{equation}
is the unique monic polynomial of the smallest degree that
satisfies (\ref{eqn:efEx}).

A closely related notion is the \emph{error locator polynomial}
\begin{equation}
      \Lambda_e(x)\eqdef \prod_{i:e_i\neq 0}m_i(x),       \label{eqn:el}
\end{equation}
which is of degree $\deg\Lambda_e(x)=\wD(e)$.
Note that $\Lambda_e(x)$ qualifies as an error factor polynomial.
In the special case where all the moduli $m_i(x), 0\leq i \leq
n-1,$ are irreducible (e.g., for irreducible polynomial remainder
codes), we have
\begin{equation}
\gcd\big(E(x),M_n(x)\big) = \prod_{i: e_i=0} m_i(x)
\end{equation}
and thus $\Lambda_f(x) = \Lambda_e(x)$.

In any case,
every error factor polynomial $\Lambda(x)$ is a multiple of $\Lambda_f(x)$.
This applies, in particular, to $\Lambda_e(x)$ and thus
\begin{equation} \label{eqn:LambdafLambdaeWd}
\deg\Lambda_f(x)\leq \deg\Lambda_e(x) = \wD(e).
\end{equation}
The following theorem is then obvious:
\begin{theorem}[Key Equation]\label{theorem:Key Equation}
The error factor polynomial (\ref{eqn:ef}) satisfies
\begin{equation}
      A(x)M_n(x)=\Lambda_f(x)E(x)         \label{eqn:KeyEqu.1}
\end{equation}
for some polynomial $A(x)\in F[x]$ of degree smaller than $\deg
\Lambda_f(x)$. Conversely, if some monic polynomial $G(x)\in F[x]$
satisfies
\begin{equation}
      A(x)M_n(x)=G(x)E(x)                \label{eqn:KeyEqu.2}
\end{equation}
for some $A(x)\in F[x]$, then $G(x)$ is a multiple of $\Lambda_f(x)$.
\end{theorem}

\noindent
For irreducible polynomial remainder codes, 
$\Lambda_f(x)$ in Theorem \ref{theorem:Key Equation} can be
replaced everywhere by $\Lambda_e(x)$ because, in this case, $\Lambda_f(x) = \Lambda_e(x)$.

The following theorem is a slight generalization of Theorem
\ref{theorem:Erasure Correction Bound}.
\begin{theorem}[Error Factor-based Interpolation] \label{theorem:Erasure Correction}
If $G(x)$ is a multiple of $\Lambda_f(x)$ with
\begin{equation}
   \deg G(x)\leq N-K, \label{eqn:con.0}
\end{equation}
then
\begin{equation} \label{eqn:LocatorBasedInterpolation}
 a(x)=\frac{G(x)Y(x) \bmod M_n(x)}{G(x)}
\end{equation}
\eproofnegspace
\end{theorem}

\begin{proof}
With $Y(x)=a(x)+E(x)$ and with $G(x)$ satisfying (\ref{eqn:con.0}), we have
\begin{eqnarray}
 G(x)Y(x) \bmod M_n(x)
    &=& G(x)\left(a(x)+E(x)\right) \bmod M_n(x) \nonumber \\
    &=& G(x)a(x)+ \tilde{E}(x)    \label{key.2}
\end{eqnarray}
with
\begin{equation}
 \tilde{E}(x)\eqdef G(x)E(x) \bmod M_n(x).  \label{key.3}
\end{equation}
If $G(x)$ is a multiple of $\Lambda_f(x)$, then $\tilde{E}(x)=0$
by Theorem~\ref{theorem:Key Equation} and
(\ref{eqn:LocatorBasedInterpolation}) follows.
\end{proof}
For irreducible polynomial remainder codes,
$\Lambda_f(x)$ in Theorem \ref{theorem:Erasure Correction} can be
replaced by $\Lambda_e(x)$
and Theorem~\ref{theorem:Erasure Correction} reduces
to Theorem \ref{theorem:Erasure Correction Bound}.
For non-irreducible codes, however,
Theorem~\ref{theorem:Erasure Correction} is more general
than Theorem~\ref{theorem:Erasure Correction Bound}
because error patterns with $\wD(e)> N-K$ but $\deg \Lambda_f(x)\leq N-K$ can exist.

\subsection{Error Factor Test and Error Locator Test}
 \label{section:Error Factor Test}

Recall $\tD \eqdef \left\lfloor \frac{N-K}{2} \right\rfloor$ from
(\ref{Basicbound:dminD}) and $\tH \eqdef \left\lfloor
\frac{n-k}{2} \right\rfloor$ from (\ref{Basicbound:dminH}).

\begin{theorem}[Error Factor Test]\label{theorem:Error Factor Test}
Let $y = \psi(a) + e$ as above,
let $G(x)$ be a nonzero polynomial,
and let
\begin{equation}
Z(x)\eqdef  G(x)Y(x) \bmod M_n(x). \nonumber
\end{equation}
Assume that the following conditions are satisfied:
\begin{enumerate}
\item $\deg \Lambda_f(x) \leq \tD$

\item $\deg G(x)\leq \tD$

\item $G(x)$ divides $Z(x)$

\item  $\deg Z(x)-\deg G(x) < K$.
\end{enumerate}
Then $G(x)$ is a multiple of $\Lambda_f(x)$ and $Z(x)=G(x)a(x)$.
\end{theorem}
Note that the conditions in the theorem are satisfied for $G(x)=\Lambda_f(x)$.
Note also that for non-irreducible polynomial remainder codes, there
may exist error patterns such that $\wD(e) > \tD$ but $\deg
\Lambda_f(x) \leq \tD$. For irreducible polynomial remainder
codes, Condition~1 in Theorem~\ref{theorem:Error Factor Test} is
equivalent to $\deg \Lambda_e(x)=\wD(e)\leq \tD$, and
$\Lambda_f(x)$ in Theorem~\ref{theorem:Error Factor Test} can be
replaced everywhere by $\Lambda_e(x)$.

\begin{proofof}{of Theorem~\ref{theorem:Error Factor Test}}
Assume that Conditions 1--4 are satisfied. Note that Condition~2
implies (\ref{eqn:con.0}), and thus (\ref{key.2})
and~(\ref{key.3}). From (\ref{key.2}) and Condition~3, we have
\begin{equation}
      \tilde{E}(x)=G(x)Q(x)      \label{key.3.1}
\end{equation}
for some polynomial $Q(x)$
and (\ref{key.2}) can be written as
\begin{equation}
      Z(x)=G(x)(a(x)+Q(x)).      \label{key.4}
\end{equation}
From Condition 4, we then have
\begin{equation}
\deg Q(x)< K.     \label{key.4.1}
\end{equation}
Furthermore, from (\ref{key.3}) and (\ref{key.3.1}), we have
$G(x)E(x)=b(x)M_n(x)+G(x)Q(x)$
for some polynomial $b(x)$ and thus
\begin{equation}
      G(x)\left(E(x)-Q(x)\right)=b(x)M_n(x).  \label{key.5}
\end{equation}
Let
\begin{equation}
\overline{\Lambda}_f(x) \eqdef M_n(x)/\Lambda_f(x)
= \gcd\!\big( E(x), M_n(x) \big).
\end{equation}
Since $\deg\Lambda_f(x)\leq \tD$,
we have $\deg \overline{\Lambda}_f(x) \geq N-\tD$.
Taking~(\ref{key.5}) modulo $\overline{\Lambda}_f(x)$ yields
\begin{equation}
G(x)Q(x) \bmod \overline{\Lambda}_f(x)=0   \label{key.6}
\end{equation}
since $E(x) \bmod \overline{\Lambda}_f(x)=0$. From (\ref{key.6}),
we have either $Q(x)=0$ or $\deg Q(x)\geq \deg
\overline{\Lambda}_f(x)-\deg G(x)\geq N-2\tD\geq K$ since $\deg
G(x)\leq \tD$. From (\ref{key.4.1}), we then conclude $Q(x)=0$.
Thus $\tilde{E}(x)=0$ from (\ref{key.3.1}) and $Z(x)=G(x)a(x)$
from (\ref{key.2}). Finally, from (\ref{key.3}) (with
$\tilde{E}(x)=0$) and the converse part of Theorem
\ref{theorem:Key Equation}, it follows that $G(x)$ is a multiple
of $\Lambda_f(x)$.
\end{proofof}

If the code $C$ further satisfies the Ordered-Degree Condition
(\ref{eqn:OrderedDegreeCondition}), we have the following analog
of Theorem \ref{theorem:Error Factor Test}. Let $N_\text{zero}(G)$
denote the number of indices $j \in \{0,\ldots,n-1\}$ such that $G(x)
\bmod m_j(x)=0$. Note that $N_\text{zero}(\Lambda_e)=\wH(e)$.

\begin{theorem}[Error Locator Test]\label{theorem:Error Locator Test}
Let $C$ be a polynomial remainder code that satisfies the
Ordered-Degree Condition and let $y = \psi(a) + e$ as above. For
some set $S\subset \{ 0, 1, \ldots, n-1 \}$ of indices, let
$G(x)=\prod_{i\in S} m_i(x)\neq 0$ and let
\begin{equation}
Z(x)\eqdef  G(x)Y(x) \bmod M_n(x). \nonumber
\end{equation}
Assume that the following conditions are satisfied:

\begin{enumerate}

\item $\wH(e) \leq \tH$

\item $N_\text{zero}(G) \leq \tH$ and $\deg G(x)\leq
\sum_{i=n-\tH}^{n-1}\deg m_i(x)$

\item $G(x)$ divides $Z(x)$

\item $\deg Z(x)-\deg G(x) < K$.

\end{enumerate}
Then, $G(x)$ is a multiple of $\Lambda_e(x)$ and $Z(x)=G(x)a(x)$.
\end{theorem}

\noindent Note that the conditions in the theorem are satisfied
for $G(x)=\Lambda_e(x)$.

\begin{proof}
Note that Condition~2 implies
(\ref{eqn:con.0}) and Conditions 3 and~4 are the same as the two
corresponding conditions in Theorem \ref{theorem:Error Factor Test}.
Assume now that Conditions 1--4 are satisfied.
It is easily verified that we then have
both (\ref{key.2})--(\ref{key.3}) and (\ref{key.3.1})--(\ref{key.5})
for some polynomial $Q(x)$.
Let $S_\text{zero}$ denote the set of indices $i \in \{0,1,\ldots,n-1\}$
such that $E(x) \bmod m_i(x) = 0$. Equation (\ref{key.5}) implies
that, for each $i \in S_\text{zero}$, we have
\begin{equation}
G(x)Q(x) \bmod m_i(x)=0   \label{key.6.1}
\end{equation}
and thus
$N_\text{zero}(Q)\geq |S_\text{zero}|-N_\text{zero}(G)$.
Since $N_\text{zero}(G) \leq \tH$
and $|S_\text{zero}|=n-\wH(e) \geq n-\tH$,
we have $N_\text{zero}(Q) \geq n-2\tH$.
It follows that $N_\text{zero}(Q) \geq k$,
which implies either $\deg Q(x)\geq K$ or $Q(x)=0$.
It then follows from (\ref{key.4.1}) that $Q(x)=0$.

We then have $\tilde{E}(x)=0$ from (\ref{key.3.1}) and thus
$Z(x)=G(x)a(x)$ from (\ref{key.2}). Finally, from (\ref{key.3})
(with $\tilde{E}(x)=0$) and the converse part of Theorem
\ref{theorem:Key Equation}, it follows that $G(x)$ $(=\prod_{i\in
S} m_i(x))$ is a multiple of $\Lambda_e(x)$.
\end{proof}

\section{Decoding by the Extended GCD Algorithm}
\label{section:exgcd}

For Reed-Solomon codes, the use of the extended gcd algorithm to
compute an error locator polynomial is standard
\cite{Sugiyama,Roth}. Gcd-based decoding of polynomial remainder
codes was proposed by Shiozaki \cite{Shiozaki}. However, the
assumptions in \cite{Shiozaki} do not cover all codes considered
in the present paper. In particular, in \cite{Shiozaki}, the
moduli $m_i(x)$ are assumed to have the same degree and they are
implicitly assumed to be irreducible, as will be discussed in
Section~\ref{sec:PriorGCDDecoding}. In order to properly address
these issues, we need to develop gcd-based decoding accordingly.
We then obtain several versions of gcd-based decoding (summarized
in Section~\ref{section:DecodingSummary}), some of which are not
quite standard even when specialized to Reed-Solomon codes.

\subsection{An Extended GCD Algorithm}
\label{section:Extended GCD Algorithm}

As in Section~\ref{section: Locator Factor Key Equation}, let $c$
be the transmitted codeword, let $e$ be the error pattern, and let
$y=c+e$ be the corrupted codeword that the receiver gets to see.
Let $a(x)$, $E(x) = \sum_{\ell=0}^{N-1} E_\ell\, x^{\ell}$, and
$Y(x) = \sum_{\ell=0}^{N-1} Y_\ell\, x^{\ell}$ be the pre-images
of these quantities with respect to $\psi$. The general idea of
gcd decoding is to compute $\gcd\big(M_n(x), E(x)\big)$ despite
the fact that $E(x)$ is not fully known. We begin by stating the
extended gcd algorithm in the following (not quite standard) form,
where we assume for the moment that $E(x)$ is fully known.

\begin{trivlist}
\item{} \noindent
{\bf Extended GCD Algorithm} \\
Input: $M_n(x)$ and $E(x)$ with $\deg M_n(x) > \deg E(x)$. \\
Output: polynomials \mbox{$\tilde{r}(x), s(x), t(x) \in F[x]$}
where $\tilde{r}(x)=\gamma  \, \text{gcd}\big(M_n(x),E(x)\big)$ for some
nonzero $\gamma \in F$ and where $s(x)$ and $t(x)$ satisfy $s(x)\cdot
M_n(x)+t(x)\cdot E(x)=0$.
\begin{pseudocode}
\npcl[gcdline:NoErrorBegin]   \pkw{if} $E(x)=0$ \pkw{begin}\\
\npcl[gcdline:NoErrorAss]   \>  $\tilde{r}(x):=M_n(x)$, $s(x):=0$, $t(x):=1$ \\
\npcl[gcdline:NoErrorReturn] \> \pkw{return} $\tilde{r}(x)$, $s(x)$, $t(x)$ \\
\npcl[gcdline:NoErrorEnd]   \pkw{end} \\
\npcl[gcdline:Initrx]   $r(x) := M_n(x)$ \\
\npcl[gcdline:Line2]   $\tilde{r}(x) := E(x)$ \\
\npcl[gcdline:Lines]   $s(x) := 1$ \\
\npcl   $t(x) := 0$ \\
\npcl[gcdline:LineTildes]  $\tilde{s}(x) := 0$ \\
\npcl[gcdline:LastInit]   $\tilde{t}(x) := 1$ \\
\npcl[gcdline:BeginLoop]  \pkw{loop begin} \\
\npcl   \> $i := \deg r(x)$ \\
\npcl[gcdline:Assignj]   \> $j := \deg \tilde{r}(x)$ \\
\npcl[gcdline:BeginWhile]   \> \pkw{while} $i \geq j$ \pkw{begin}\\
\npcl[gcdline:ri]   \> \> $q(x):=\frac{r_i}{\tilde{r}_j}~x^{i-j}$ \\
\npcl[gcdline:Updater]   \> \> $r(x):=r(x)-q(x)\cdot \tilde{r}(x)$ \\
\npcl[gcdline:Updates]   \> \> $s(x):=s(x)-q(x)\cdot \tilde{s}(x)$ \\
\npcl[gcdline:Updatet]   \> \> $t(x):=t(x)-q(x)\cdot \tilde{t}(x)$ \\
\npcl[gcdline:i]   \> \> $i :=\deg r(x)$ \\
\npcl[gcdline:EndWhile]   \> \pkw{end} \\
\npcl[gcdline:BeginStoppingIf]   \> \pkw{if} $r(x)=0$ \pkw{begin}\\
\npcl[gcdline:Return]   \> \> \pkw{return} $\tilde{r}(x)$, $s(x)$, $t(x)$ \\
\npcl[gcdline:BeginStoppingEnd]   \> \pkw{end} \\
\npcl[gcdline:Swapr]   \> $(r(x),\tilde{r}(x)) := (\tilde{r}(x),r(x))$ \\
\npcl[gcdline:Swaps]   \> $(s(x),\tilde{s}(x)):=(\tilde{s}(x),s(x))$ \\
\npcl[gcdline:Swapt]   \> $(t(x),\tilde{t}(x)):=(\tilde{t}(x),t(x))$ \\
\npcl   \pkw{end}   \\
\end{pseudocode}
\vspace{-6ex} \hfill $\Box$
\end{trivlist}
The inner loop between lines \ref{gcdline:BeginWhile}
and~\ref{gcdline:EndWhile} essentially computes the division of
$r(x)$ by $\tilde{r}(x)$. In line~\ref{gcdline:ri}, $r_i$ denotes
the coefficient of $x^i$ in $r(x)$ and $\tilde{r}_j$ denotes the
coefficient of $x^j$ in $\tilde{r}(x)$. For polynomials over
$F=\text{GF}(2)$, the scalar division $r_i/\tilde{r}_j$ in
line~\ref{gcdline:ri} disappears.

\begin{theorem}[GCD Loop Invariants]\label{theorem:GCD Loop Invariant}
The condition
\begin{equation} \label{eqn:GCDMainInvariant}
\gcd\big(M_n(x), E(x)\big) = \gcd\big( r(x), \tilde{r}(x) \big)
\end{equation}
holds everywhere after line~\ref{gcdline:Line2}.
The condition
\begin{equation}
r(x) = s(x)\cdot M_n(x)+t(x)\cdot E(x) \label{eqn:XGCDLoopInv1}
\end{equation}
holds 
both between lines \ref{gcdline:Assignj} and~\ref{gcdline:BeginWhile}
and between lines \ref{gcdline:EndWhile} and~\ref{gcdline:BeginStoppingIf}.
The condition
\begin{equation}
\deg M_n(x) = \deg \tilde{r}(x)+ \deg t(x) \label{degcheck}
\end{equation}
holds between lines \ref{gcdline:EndWhile}
and~\ref{gcdline:BeginStoppingIf}.
\end{theorem}
Equations\ (\ref{eqn:GCDMainInvariant})
and~(\ref{eqn:XGCDLoopInv1}) are the standard loop invariants of
extended gcd algorithms, cf.\ e.g.\ \cite{Roth}. The proof of
Theorem~\ref{theorem:GCD Loop Invariant} is given in
Appendix~B.

\begin{theorem}[GCD Output] \label{theorem:GCD Output}
When the algorithm terminates, 
we have both
\begin{eqnarray}
\tilde{r}(x)&=&\gamma~ \text{gcd}\big(M_n(x),E(x)\big)  \label{eqn:xgcdMainOutput}\\
     &=&\gamma~ \frac{M_n(x)}{\Lambda_f(x)}  \label{eqn:xgcdrtildeMnLambda}
\end{eqnarray}
for some nonzero $\gamma \in F$
and
\begin{equation} \label{eqn:gcdtLambdaf}
t(x) = \tilde{\gamma} \Lambda_f(x)
\end{equation}
for some nonzero $\tilde{\gamma} \in F$. Moreover, the returned
$s(x)$ and $t(x)$ satisfy
\begin{equation}
s(x)\cdot M_n(x)+t(x)\cdot E(x)=0.  \label{eqn:sMtE}
\end{equation}
\eproofnegspace
\end{theorem}

\begin{proof}
If $E(x)=0$, the algorithm terminates at line~\ref{gcdline:NoErrorReturn}
and (\ref{eqn:xgcdMainOutput})--(\ref{eqn:sMtE}) are easily verified.

We now prove the case where $E(x)\neq 0$. Equation
(\ref{eqn:xgcdMainOutput}) follows from
(\ref{eqn:GCDMainInvariant}) and (\ref{eqn:xgcdrtildeMnLambda})
follows from (\ref{eqn:ef}). It remains to prove
(\ref{eqn:gcdtLambdaf}) and (\ref{eqn:sMtE}). With $r(x)=0$  and
from (\ref{eqn:XGCDLoopInv1}), Equation (\ref{eqn:sMtE}) follows.
We then conclude from the second part of Theorem~\ref{theorem:Key
Equation} that $t(x)$ is a multiple of $\Lambda_f(x)$. Finally, it
follows from (\ref{degcheck}) and (\ref{eqn:xgcdrtildeMnLambda})
that $t(x)$ and $\Lambda_f(x)$ have the same degree.
\end{proof}

From (\ref{eqn:gcdtLambdaf}), we see that the gcd algorithm
computes the error factor polynomial $\Lambda_f$ (up to a scale
factor). The main idea of gcd decoding (discovered by Sugiyama
\cite{Sugiyama}) is that this still works even if $E(x)$ is only
partially known.

\subsection{Modifications for Partially Known $E(x)$}
\label{section:Modified GCD Algorithm}

Recall that $Y(x)=a(x)+E(x)$ where $E(x) = \sum_{\ell=0}^{N-1}
E_\ell\, x^{\ell}$ is the pre-image of $e$. Since $\deg a(x) < K$,
the receiver knows the coefficients $E_{K}, E_{K+1},\ldots,
E_{N-1}$ of $E(x)$, but not $E_0,\ldots, E_{K-1}$. With the
following modifications, the Extended GCD Algorithm of
Section~\ref{section:Extended GCD Algorithm}
can still be used to compute (\ref{eqn:gcdtLambdaf}). 

\begin{trivlist}
\item{} \noindent \noindent {\bf Partial GCD Algorithm~I} \\
Input: $M_n(x)$ and $Y(x)$ with $\deg M_n(x) > \deg Y(x)$. \\
Output: $r(x)$, $s(x)$ and $t(x)$, cf.\
Theorem~\ref{theorem:AlgorithmI Output} below.

The algorithm is the same as the Extended GCD Algorithm of
Section~\ref{section:Extended GCD Algorithm} except for the
following changes:

\begin{itemize}
\item Line~\ref{gcdline:NoErrorBegin}:
  \textbf{if} $\deg Y(x) < K$ \textbf{begin}

\item Line~\ref{gcdline:NoErrorAss}:
   $r(x):=Y(x)$, $s(x):=0$, $t(x):=1$

\item Line~\ref{gcdline:Line2}: $\tilde{r}(x) :=Y(x)$

\item Line \ref{gcdline:BeginStoppingIf}:
  \begin{equation} \label{eqn:ModifiedStopping1}
  \text{\textbf{if} $\deg r(x) < \deg t(x)+K$ \textbf{begin}}
  \end{equation}
  or alternatively
  \begin{equation} \label{eqn:ModifiedStopping2}
  \text{\textbf{if} $\deg r(x) <  (N+K)/2$ \textbf{begin}}
  \end{equation}
\end{itemize}
\vspace{-2ex} \hfill $\Box$
\end{trivlist}

\begin{theorem}\label{theorem:AlgorithmI Output}
If
\begin{equation}
 \deg \Lambda_f(x)\leq (N-K)/2,   \label{d_bound.0}
\end{equation}
then the Partial GCD Algorithm~I (with either
(\ref{eqn:ModifiedStopping1}) or (\ref{eqn:ModifiedStopping2}))
returns the same polynomials $s(x)$ and $t(x)$ (after the same
number of iterations) as the Extended GCD Algorithm of
Section~\ref{section:Extended GCD Algorithm}. Moreover, the
returned $r(x)$ is such that
\begin{equation} \label{eqn:ModifiedGCDI_rx}
r(x)=t(x)a(x).
\end{equation}
\eproofnegspace
\end{theorem}
The proof is given in Appendix~B.
Note that $a(x)$ can be recovered directly from
(\ref{eqn:ModifiedGCDI_rx}).

\subsection{Alternative Modifications for Partially Known $E(x)$}
\label{section:Modified GCD Algorithm II}

The Partial GCD Algorithm~I of the previous section involves a lot
of computations with the unknown lower parts of $E(x)$. These
computations are avoided in the following algorithm, which works
only with the known part of $E(x)$ as follows.
Let
\begin{equation} \label{EU}
E_U(x) \eqdef \sum_{\ell=0}^{N-K-1} E_{K+\ell} \, x^\ell
= \sum_{\ell=0}^{N-K-1} Y_{K+\ell} \, x^\ell,
\end{equation}
which is the known upper part of $E(x)= \sum_{\ell=0}^{N-1}
E_\ell\, x^{\ell}$, and let
\begin{equation}
M_U(x) \eqdef \sum_{\ell=0}^{N-K}(M_n)_{K+\ell} \, x^\ell
\label{MU}
\end{equation}
be the corresponding upper part of
$M_n(x)=\sum_{\ell=0}^{N} (M_n)_\ell\, x^{\ell}$.

\begin{trivlist}
\item{} \noindent \noindent {\bf Partial GCD Algorithm~II } \\
Input: $M_U(x)$ and $E_U(x)$ with $\deg M_U(x) > \deg E_U(x)$. \\
Output: $s(x)$ and $t(x)$, cf.\ Theorem~\ref{theorem:AlgorithmII
Output} below.

The algorithm is the same as the Extended GCD Algorithm of
Section~\ref{section:Extended GCD Algorithm} except for the
following changes:

\begin{itemize}
\item Line~\ref{gcdline:NoErrorBegin}:
  \textbf{if} $E_U(x)=0$ \textbf{begin}

\item Line~\ref{gcdline:NoErrorAss}:
   $s(x):=0$, $t(x):=1$

\item Line~\ref{gcdline:Initrx}: $r(x) := M_U(x)$
\item Line~\ref{gcdline:Line2}: $\tilde{r}(x) :=E_U(x)$

\item Line \ref{gcdline:BeginStoppingIf}:
  \begin{equation} \label{eqn:ModifiedStrippedStopping1}
   \text{\textbf{if} $\deg r(x) < \deg t(x)$ \textbf{begin}}
  \end{equation}
  or alternatively
  \begin{equation} \label{eqn:ModifiedStrippedStopping2}
   \text{\textbf{if} $\deg r(x) <  (N-K)/2$ \textbf{begin}}
  \end{equation}
\end{itemize}
\vspace{-2ex} \hfill $\Box$
\end{trivlist}

\begin{theorem}\label{theorem:AlgorithmII Output}
If the condition (\ref{d_bound.0}) is satisfied, then the Partial
GCD Algorithm~II (with either
(\ref{eqn:ModifiedStrippedStopping1}) or
(\ref{eqn:ModifiedStrippedStopping2})) returns the same
polynomials $s(x)$ and $t(x)$ (after the same number of
iterations) as the Extended GCD Algorithm of
Section~\ref{section:Extended GCD Algorithm}.
\end{theorem}
The proof is given in Appendix~C. Note, however, that this
algorithm does not compute $r(x)$ as in
(\ref{eqn:ModifiedGCDI_rx}).

\subsection{Summary of Decoding}
\label{section:DecodingSummary}

We can now put together several decoding algorithms
that consist of the following three steps.
The relation of all these decoding algorithms to the prior
literature is discussed in Section~\ref{sec:PriorGCDDecoding}.

\begin{enumerate}
\item \textbf{Transform:} Compute $Y(x) = \psi^{-1}(y)$. If $\deg
Y(x)<K$, we conclude $E(x)=0$ and $a(x) = Y(x)$, and the following
two steps can be skipped.

\item \textbf{Partial GCD:} If $\deg Y(x) \geq K$, run either the
Partial GCD Algorithm~I (Section \ref{section:Modified GCD
Algorithm}) or the Partial GCD Algorithm~II (Section
\ref{section:Modified GCD Algorithm II}). Either algorithm yields
the polynomial $t(x) = \tilde{\gamma} \Lambda_f(x)$ (for some
scalar $\tilde{\gamma}\in F$) provided that $\deg \Lambda_f(x)
\leq (N-K)/2$.

If $\deg t(x) > (N-K)/2$, we declare a decoding failure.

Depending on Step~3 (below),
the computation of the polynomials $s(x)$ and $\tilde{s}(x)$
may be unnecessary. In this case,
lines \ref{gcdline:Lines},
\ref{gcdline:LineTildes}, \ref{gcdline:Updates}, and
\ref{gcdline:Swaps} of the gcd algorithm can be deleted.

\item \textbf{Recovery:} Recover $a(x)$ by any of the following
methods:
\begin{enumerate}

\item From (\ref{eqn:LocatorBasedInterpolation}), we have
\begin{equation} \label{eqn:ax from tx}
 a(x)=\frac{t(x)Y(x) \bmod M_n(x)}{t(x)}
\end{equation}

(If the numerator of (\ref{eqn:ax from tx}) is not a multiple of
$t(x)$ or if $\deg a(x) \geq K$, then decoding failed due to some
uncorrectable error.)

\item When using the Partial GCD Algorithm~I in the Step~2, we can
compute $a(x) = r(x)/t(x)$ according to
(\ref{eqn:ModifiedGCDI_rx}).

(If $t(x)$ does not divide $r(x)$ or if $\deg a(x) \geq K$, we
declare a decoding failure.)

\item Alternatively, from (\ref{eqn:sMtE}), we can compute
\begin{equation}\label{eqn:computeE(x)}
 E(x)=\frac{-s(x)\cdot M_n(x)}{t(x)}
\end{equation}
and then obtain $a(x) = Y(x) - E(x)$.

(If the numerator of (\ref{eqn:computeE(x)}) is not a multiple of
$t(x)$ or if $\deg a(x) \geq K$, we declare a decoding failure.)

The computation can be simplified as follows. Let $E_L(x)\eqdef
E(x)-x^KE_U(x)$ denote the unknown part of $E(x)$. Then
\begin{equation}\label{eqn:computeEL(x)}
 E_L(x)=\frac{-s(x)\cdot M_n(x)- x^Kt(x)E_U(x)}{t(x)}
\end{equation}
and $a(x)$ can be recovered by $a(x)=\sum_{\ell=0}^{K-1} Y_\ell \,
x^{\ell}-E_L(x)$.

\end{enumerate}
\end{enumerate}

As stated, the described decoding algorithms are guaranteed to
correct all errors $e$ with $\deg \Lambda_f(x) \leq \tD$, which by
(\ref{eqn:LambdafLambdaeWd}) implies that they also correct all
errors $e$ with $\wD(e) \leq \tD$ (\ref{Basicbound:dminD}). If the
code satisfies the Ordered-Degree Condition
(\ref{eqn:OrderedDegreeCondition}) as well as the additional
condition
\begin{equation} \label{sim_setting}
\deg m_{k}(x)= \cdots = \deg m_{n-1}(x),
\end{equation}
then the algorithm is guaranteed to correct also all errors $e$
with $\wH(e) \leq \tH$ (\ref{Basicbound:dminH}) since in this
case, from (\ref{eqn:el}),
$\wH(e) \leq \tH$ implies $\wD(e) \leq \tD$.

\subsubsection*{An Extension}

Assume that the code satisfies the Ordered-Degree Condition
(\ref{eqn:OrderedDegreeCondition}) but not the additional
condition (\ref{sim_setting}). In this case, we can still correct
all errors $e$ with $\wH(e) \leq \tH$
(in addition to all errors with $\wD(e) \leq \tD$)
by the following procedure,
which, however, is practical only in special cases.

\begin{trivlist}
\item{} \noindent
{\bf Decoder with List of Special Error Positions}\\
First, run the gcd decoder of the previous section. If it
succeeds, stop. Otherwise, let $S_\Lambda$ be a precomputed list
of candidate error locator polynomials $G(x)$ with
$N_\text{zero}(G) \leq \tH$ and $\deg G(x) > (N-K)/2$. Check if
any $G(x) \in S_\Lambda$ satisfies all conditions of
Theorem~\ref{theorem:Error Locator Test}. If such a polynomial
$G(x)$ exists, we conclude that it is a multiple of the error
locator polynomial and we compute $a(x)$ from
(\ref{eqn:LocatorBasedInterpolation}). \hfill $\Box$
\end{trivlist}

\subsection{Relation to Prior Work}
\label{sec:PriorGCDDecoding}

The idea of gcd-based decoding is due to Sugiyama \cite{Sugiyama}
and its application to polynomial remainder codes is due to
Shiozaki \cite{Shiozaki}. As it turns out, most (and perhaps all)
gcd-based decoding algorithms in the literature, both for
Reed-Solomon codes and for polynomial residue codes, are
essentially identical to one of the algorithms of
Section~\ref{section:DecodingSummary}. However, even when
specialized to Reed-Solomon codes, no single paper (not even
\cite{Fedorenko,Fedorenko2}) seems to cover all these algorithms.
In particular, recovering $a(x)$ by (\ref{eqn:ax from tx}) does
not seem to have appeared in the literature. For Reed-Solomon
codes, the work by Gao \cite{Gao} appears to be the most
pertinent, see also \cite{Fedorenko,Fedorenko2}.
As for polynomial remainder codes,
our algorithms overcome the limitations of
Shiozaki's algorithm \cite{Shiozaki} as will be discussed below.

\subsubsection*{Relation to Gao's Decoding Algorithms for Reed-Solomon Codes}

In the same paper \cite{Gao} from 2003, Gao proposed two
algorithms for decoding Reed-Solomon codes. Each algorithm
comprises three steps, and the first step of each algorithm is
essentially Step~1 (``Transform'') of Section
\ref{section:DecodingSummary}.

\begin{description}
\item[Gao's first algorithm:] Step~2 of this algorithm is
essentially the Partial GCD Algorithm~I of
Section~\ref{section:Modified GCD Algorithm} with
(\ref{eqn:ModifiedStopping2}) as the stopping condition. Step~3 is
identical to Step~3.b in Section~\ref{section:DecodingSummary}.

As pointed out in \cite{Fedorenko2},
this algorithm is actually identical to Shiozaki's 1988 algorithm for
decoding Reed-Solomon codes \cite{Shiozaki}.

\item[Gao's second algorithm:]
The stopping condition of the gcd-algorithm (Step~2)
as stated in \cite{Gao} is not quite correct:
it should be changed from $\deg g(x)<(d+1)/2$ to $\deg g(x)<(d-1)/2$
where $d \eqdef n-k+1$ is the minimum Hamming distance of the code.

With this correction, Step~2 of this algorithm is identical to the
Partial GCD Algorithm~II of Section~\ref{section:Modified GCD
Algorithm} with (\ref{eqn:ModifiedStrippedStopping2}) as the
stopping condition. Step~3 of the algorithm turns out to be
equivalent to the first part of 3.c in
Section~\ref{section:DecodingSummary}, i.e., computing
$a(x)=Y(x)-E(x)$ with $E(x)$ as in~(\ref{eqn:computeE(x)}).
\end{description}

\subsubsection*{Relation to Shiozaki's Decoding Algorithms}

In \cite{Shiozaki}, Shiozaki proposed a new version of gcd-based decoding
for Reed-Solomon codes, which he also extended to polynomial remainder codes.
(For Reed-Solomon codes, Shiozaki's algorithm is equivalent to Gao's
first decoding algorithm, as noted above.)

Shiozaki's algorithm also consists of three steps: the first step
agrees with Step~1 in Section~\ref{section:DecodingSummary}, the
second step is equivalent to the Partial GCD Algorithm~I with
(\ref{eqn:ModifiedStopping2}) as the stopping condition, and the
third step is identical to Step~3.b of
Section~\ref{section:DecodingSummary}).

However, the assumptions in \cite{Shiozaki}
do not cover all codes considered in the present paper.
First, it is assumed in \cite{Shiozaki} that
all the moduli $m_i(x)$, $0\leq i\leq n-1$,
have the same degree.

Second, the argument given in \cite{Shiozaki} seems to assume that
all the moduli are irreducible although this assumption is not
stated explicitly.
Specifically, Shiozaki derived a congruence (see (37) in
\cite{Shiozaki}) involving an error locator polynomial as defined
in (\ref{eqn:el}), and then used the gcd-based decoding algorithm
to solve the congruence. However, if the moduli are not
irreducible, then the gcd-based decoding algorithm will find an
error factor polynomial $(\ref{eqn:ef})$ (as shown in our Theorems
\ref{theorem:GCD Output} and \ref{theorem:AlgorithmI Output})
rather than an error locator polynomial.

\section{Conclusion} \label{section:Conclusion}

We considered polynomial remainder codes and their decoding
more carefully than in previous work.
We explicitly allowed the code symbols to be
polynomials of different degrees,
which leads to two different notions of weight and distance
and, correspondingly, to two different Singleton bounds.

Our discussion of algebraic decoding revolved around the notion
of an error factor polynomial, which is a generalization
of an error locator polynomial.
From a correct error factor polynomial,
the transmitted codeword can be recovered in various ways,
including a new method
for erasures-only decoding of general Chinese remainder codes.

Error factor polynomials can be computed by a suitably adapted
partial gcd algorithm. We obtained several versions of such
decoding algorithms, which generalize previous work and which
include the published gcd-based decoders of Reed-Solomon codes as
special cases.

\clearpage

\appendix

\section*{Appendix A: The Number of Monic Irreducible Polynomials}

The number of monic irreducible polynomials of any degree over any
finite field can be expressed in closed form \cite{Roth}. However,
this closed-form expression is not easy to evaluate. Therefore,
for the convenience of the reader, we tabulate some of these
numbers.

The first table gives the number $N_i$ of binary irreducible polynomials
of degree~$i$:

\begin{center}
  \begin{tabular}[t]{c||c|c|c|c|c|c|c|c|c|c|c|c}
  $i$    & 1 & 2 & 3 & 4 & 5 & 6 & 7  & 8   & 9  & 10 & 11 & 12\\
  \hline
  $N_i$  & 2 & 1 & 2 & 3 & 6 & 9 & 18 & 30 & 56 & 99 & 186 & 335 \\
  \hline
  $S_i$  & 2 & 4 & 10 & 22 & 52 & 106 & 232 & 472 & 976 & 1966 & 4012 &
  8032 \\
  \end{tabular}
\end{center}
\vspace{3 mm}
\begin{center}
  \begin{tabular}[t]{c||c|c|c|c}
  $i$    & 13 & 14 & 15 & 16\\
  \hline
  $N_i$  & 630 & 1161 & 2182 & 4080\\
  \hline
  $S_i$  & 16222 & 32476 & 65206 & 130486\\
  \end{tabular}
\end{center}
The table also gives the number
$S_i \eqdef \sum_{\ell=1}^i \ell N_{\ell}$,
which is the maximum degree of $M_n(x)$ of a polynomial remainder code
that uses only irreducible moduli of degree at most~$i$.

The second table gives
the number $N_i$ of monic irreducible polynomials over GF$(2^j)$
of degree~$i$:
\begin{center}
  \begin{tabular}[t]{c||c|c|c|c|c|c}
  & GF($2^2$) &  GF($2^4$) & GF($2^6$) & GF($2^8$) & GF($2^{10}$) & GF($2^{12}$)\\
  \hline
  $N_1$    & 4 & 16 & 64 & 256 & 1024 & 4096 \\
  \hline
  $N_2$    & 6 & 120 & 2016 & 32640 & 523776 & 8386560 \\
  \end{tabular}
\end{center}
E.g, over $\text{GF}(2^8)$,
there are $256$ monic irreducible polynomials of degree~1 and
32640 polynomials of degree 2.

\section*{Appendix B: Proof of Theorem \ref{theorem:AlgorithmI Output}}

In this section, we first prove the loop invariant properties of
the Extended GCD Algorithm in Section~\ref{section:Extended GCD
Algorithm} and the Partial GCD Algorithm~I in Section
\ref{section:Modified GCD Algorithm}, and then proceed to prove
Theorem \ref{theorem:AlgorithmI Output}.

We begin with the Extended GCD Algorithm of
Section~\ref{section:Extended GCD Algorithm}.
In order to prove Theorem~\ref{theorem:GCD Loop Invariant},
we first recall that,
for $R=\Z$ or $R=F[x]$ for some field $F$,
\begin{equation}
\gcd\big(a, b\big) = \gcd\big( a+qb,b \big)
\label{eqn:gcdloopinvarient}
\end{equation}
for all $a,b,q \in R$, provided that $a$ and $b$ are not both
zero.
It follows that (\ref{eqn:GCDMainInvariant})
holds everywhere after line~\ref{gcdline:Line2}.

The other claims of Theorem~\ref{theorem:GCD Loop Invariant}
are covered by the following lemma.
\begin{lemma}[GCD Loop Invariant]\label{lemma:GCD Loop Invariant.1}
For the Extended GCD Algorithm in Section~\ref{section:Extended
GCD Algorithm}, the condition
\begin{equation}
r(x) = s(x)\cdot M_n(x)+t(x)\cdot E(x) \label{eqn:r0}
\end{equation}
holds both between lines \ref{gcdline:Assignj} and~\ref{gcdline:BeginWhile}
and between lines \ref{gcdline:EndWhile} and~\ref{gcdline:BeginStoppingIf}.
For the Partial GCD Algorithm~I in Section~\ref{section:Modified GCD Algorithm},
the condition
\begin{equation}
r(x) = s(x)\cdot M_n(x)+t(x)\cdot Y(x) \label{eqn:modified r0}
\end{equation}
also holds both between lines \ref{gcdline:Assignj} and~\ref{gcdline:BeginWhile}
and between lines \ref{gcdline:EndWhile} and~\ref{gcdline:BeginStoppingIf}.

For both algorithms, the conditions
\begin{eqnarray}
\deg r(x) &<& \deg \tilde{r}(x) \label{eqn:degr0r1}\\
\deg t(x) &>& \deg \tilde{t}(x) \label{eqn:degt0gt1}\\
\deg M_n(x) &=& \deg \tilde{r}(x)+ \deg t(x) \label{modified
degcheck}
\end{eqnarray}
hold between lines \ref{gcdline:EndWhile}
and~\ref{gcdline:BeginStoppingIf}.

Specifically, let $\delta_\ell$ denote the degree of $q(x)$
(line~\ref{gcdline:ri}) in the first iteration of the \pkw{while}
block (lines \ref{gcdline:BeginWhile}--\ref{gcdline:EndWhile}) of
the $\ell$-th loop iteration. Then, for the respective algorithms,
\begin{equation}  \label{eqn:trackdegtx}
 \deg t(x)=\deg\tilde{t}(x)+\delta_\ell=\sum_{v=1}^{\ell} \delta_v
\end{equation}
holds between lines~\ref{gcdline:EndWhile}
and~\ref{gcdline:BeginStoppingIf} in the $\ell$-th loop iteration.
\end{lemma}

\begin{proof}
Conditions (\ref{eqn:r0}) and (\ref{eqn:modified r0}) are loop
invariants (of the respective algorithms), as is easily verified.
Inequality (\ref{eqn:degr0r1}) is obvious. It remains to prove
(\ref{eqn:degt0gt1})--(\ref{eqn:trackdegtx}). For both algorithms,
assume the conditions
\begin{eqnarray}
 \deg r(x) &>& \deg \tilde{r}(x) \label{proof:degr0gr1}\\
 \deg t(x) &<& \deg \tilde{t}(x) \label{proof:degt0t1}\\
 \deg M_n(x)&=&\deg r(x)+ \deg \tilde{t}(x)  \label{proof:degcheck.1}
\end{eqnarray}
hold between lines \ref{gcdline:Assignj} and
\ref{gcdline:BeginWhile} in the $\ell$-th loop iteration. Note
that $r(x)$, $\tilde{r}(x)$, $t(x)$, and $\tilde{t}(x)$ are
initialized to $M_n(x)$, $E(x)$ or $Y(x)$, $0$, and $1$,
respectively; thus
(\ref{proof:degr0gr1})--(\ref{proof:degcheck.1}) obviously hold
between lines \ref{gcdline:Assignj} and \ref{gcdline:BeginWhile}
in the first iteration. In the following, we begin with $\ell=1$
and then complete the proof by induction.

For both algorithms, let $d_\ell=\deg r(x)$ denote the degree of
$r(x)$ between lines \ref{gcdline:Assignj} and
\ref{gcdline:BeginWhile} in the $\ell$-th loop iteration, and let
$\delta_\ell$ denote the degree of $q(x)$ (line~\ref{gcdline:ri})
in the first iteration of the \pkw{while} block (lines
\ref{gcdline:BeginWhile}--\ref{gcdline:EndWhile}) of the $\ell$-th
loop iteration. Note that $\delta_\ell= d_\ell-\deg \tilde{r}(x)>
0$ and from (\ref{proof:degcheck.1})
\begin{equation}
 \deg M_n(x)=d_\ell+ \deg \tilde{t}(x).  \label{degcond}
\end{equation}
Recall that, from (\ref{proof:degt0t1}), $\deg t(x) < \deg
\tilde{t}(x)$ holds before entering the \pkw{while} block, and
recall the update rule for $t(x)$ in line~\ref{gcdline:Updatet}.
Clearly, in the first execution of line~\ref{gcdline:Updatet}, the
degree of $t(x)$ is increased to $\deg \tilde{t}(x)+\delta_\ell$,
and further iterations inside the \pkw{while} block will not
change $\deg t(x)$ since $\deg q(x)$ decreases in each iteration.
It follows that $\deg t(x)=\deg \tilde{t}(x)+\delta_\ell$ holds
between lines \ref{gcdline:EndWhile}
and~\ref{gcdline:BeginStoppingIf}, and in particular, $\deg
t(x)=\delta_1$ holds when $\ell=1$ because $\deg \tilde{t}(x)=0$
holds throughout the \pkw{while} block of the first loop
iteration. Thus, (\ref{eqn:degt0gt1}) and (\ref{eqn:trackdegtx})
both hold between lines \ref{gcdline:EndWhile}
and~\ref{gcdline:BeginStoppingIf} in the first loop iteration.
Further, since $\delta_\ell= d_\ell-\deg \tilde{r}(x)$,  we have
\begin{align}
 \deg t(x) &=\deg \tilde{t}(x)+d_\ell-\deg \tilde{r}(x)  \\
           &=\deg M_n(x)-\deg \tilde{r}(x),
\end{align}
where the last step follows from (\ref{degcond}), and thus
(\ref{modified degcheck}) holds between lines
\ref{gcdline:EndWhile} and~\ref{gcdline:BeginStoppingIf} in the
$\ell$-th loop iteration.

After the swaps of the corresponding auxiliary polynomials in
lines \ref{gcdline:Swapr}--\ref{gcdline:Swapt},  the conditions
(\ref{proof:degr0gr1})--(\ref{proof:degcheck.1}) hold again
between lines \ref{gcdline:Assignj} and \ref{gcdline:BeginWhile}
for the subsequent loop iteration. In particular, for $\ell=2$,
$\deg \tilde{t}(x)=\delta_1$ holds between lines
\ref{gcdline:Assignj} and \ref{gcdline:BeginWhile} in the second
loop iteration. The proof is then completed by
induction. 
\end{proof}

We now start to prove Theorem \ref{theorem:AlgorithmI Output}. If
$E(x)=0$, which implies \mbox{$\deg Y(x)<K$}, Theorem
\ref{theorem:AlgorithmI Output} holds obviously; we thus prove in
the following only the case where $E(x)\neq 0$.
For the Partial GCD Algorithm~I in Section \ref{section:Modified
GCD Algorithm}, let $g$ denote the largest integer such that the
coefficient of $x^g$ of either $r(x)$ or of $\tilde{r}(x)$ is
unknown, or alternatively let $g$ denote the largest integer such
that the coefficient of $x^g$ of either $r(x)$ or of
$\tilde{r}(x)$ is ``probably unmatched'' with the corresponding
$r(x)$ or the corresponding $\tilde{r}(x)$ in the Extended GCD
Algorithm of Section~\ref{section:Extended GCD Algorithm} when we
run both algorithms simultaneously. Clearly, the algorithm starts
with $g = K-1$, since the coefficients $E_0,E_1,\ldots, E_{K-1}$
of $\tilde{r}(x):=Y(x)$ (line~\ref{gcdline:Line2}) are unknown.
Moreover, let $h\eqdef \max\{\deg r(x),\deg \tilde{r}(x)\}$.
Clearly, the algorithm starts with $h=\deg M_n(x)=N$.
\begin{lemma}\label{lemma:GCD Algorithm.1} For the Partial GCD Algorithm~I of Section
\ref{section:Modified GCD Algorithm}, let $\delta_\ell$ denote the
degree of $q(x)$ in the first iteration of the \pkw{while} block
(lines \ref{gcdline:BeginWhile}--\ref{gcdline:EndWhile}) of the
$\ell$-th loop iteration. If $h-g>2\delta_\ell$ holds between
lines~\ref{gcdline:Assignj} and~\ref{gcdline:BeginWhile}, then the
value of $q(x)$ (line~\ref{gcdline:ri}) throughout the \pkw{while}
block in the $\ell$-th loop iteration is exactly the same as the
corresponding one of the Extended GCD Algorithm of
Section~\ref{section:Extended GCD Algorithm} in the same loop
iteration. In addition, $g=(K-1)+\sum_{v=1}^{\ell} \delta_v$ and
$h=N-\sum_{v=1}^{\ell} \delta_v$ both hold between lines
\ref{gcdline:EndWhile} and~\ref{gcdline:BeginStoppingIf} in the
$\ell$-th loop iteration.
\end{lemma}

\begin{proof}
We will prove this theorem by induction. Recall that the update
rule for $r(x)$ in line \ref{gcdline:Updater} is
\begin{equation}
   r(x):=r(x)-q(x)\cdot \tilde{r}(x). \label{updr0}
\end{equation}
In the first loop iteration, $h=\deg r(x)=N$ and $g =K -1$ clearly
hold between lines~\ref{gcdline:Assignj}
and~\ref{gcdline:BeginWhile}, and $g$ is the largest integer such
that the coefficient of $x^g$ of $\tilde{r}(x)$ is unknown. If
$h-g>2\delta_1$ holds between lines~\ref{gcdline:Assignj}
and~\ref{gcdline:BeginWhile}, then the first execution of
(\ref{updr0}) in the \pkw{while} block increases $g$ by
$\delta_1$; afterwards, further iterations in the same block will
not change $g$ since $\deg q(x)$ decreases in each iteration.
Moreover, after executing the \pkw{while} block, $h=\deg
\tilde{r}(x)=N-\delta_1$ holds between lines
\ref{gcdline:EndWhile} and~\ref{gcdline:BeginStoppingIf}. It is
also easily seen that throughout the \pkw{while} block, the value
of $q(x)$ in line~\ref{gcdline:ri} is exactly identical to the
corresponding one of the Extended GCD Algorithm.

Note that the increased $g$, i.e., after the first execution of
(\ref{updr0}), will become to denote the largest integer such that
the coefficient of $x^g$ of $r(x)$ is unknown. It follows after
the swap of $r(x)$ and $\tilde{r}(x)$ in line \ref{gcdline:Swapr}
that the increased $g$ will again become to denote the largest
integer such that the coefficient of $x^g$ of $\tilde{r}(x)$ is
unknown between lines~\ref{gcdline:Assignj}
and~\ref{gcdline:BeginWhile} for subsequent loop iteration, and
the decreased $h$ will again become to denote $\deg r(x)$ between
lines~\ref{gcdline:Assignj} and~\ref{gcdline:BeginWhile} for
subsequent loop iteration. The proof is then completed by
induction.
\end{proof}
Since $h-g=N-K+1$ holds between lines~\ref{gcdline:Assignj}
and~\ref{gcdline:BeginWhile} in the first loop iteration, it
follows from Lemma~\ref{lemma:GCD Algorithm.1} that if
\begin{equation}\label{eqn:predbound}
2 \sum_{v=1}^{\ell} \delta_v <  N-K+1,
\end{equation}
then, from the first to the $\ell$-th loop iteration, $q(x)$ and
thus $s(x)$ and $t(x)$  are exactly the same as in the Extended
GCD Algorithm. Moreover from Lemma~\ref{lemma:GCD Loop
Invariant.1}, $\deg t(x)=\sum_{v=1}^{\ell} \delta_v$ holds between
lines \ref{gcdline:EndWhile} and~\ref{gcdline:BeginStoppingIf}. In
order to obtain (\ref{eqn:gcdtLambdaf}), which implies that $\deg
t(x)=\deg \Lambda_f(x)$, it turns out from (\ref{eqn:predbound})
that if
\begin{equation}
2\deg \Lambda_f(x)\leq N-K,
\end{equation}
which agrees with (\ref{d_bound.0}), then the algorithm maintains
exactly the same $s(x)$ and $t(x)$ as the Extended GCD Algorithm
of Section~\ref{section:Extended GCD Algorithm} until $\deg
t(x)=\deg \Lambda_f(x)$.

It remains to argue the validity of (\ref{eqn:ModifiedStopping1})
and (\ref{eqn:ModifiedStopping2}) (i.e., line
\ref{gcdline:BeginStoppingIf} in the Partial GCD Algorithm~I) as
appropriate terminating conditions. Assume now that
(\ref{d_bound.0}) is satisfied and suppose the Extended GCD
Algorithm (in Section~\ref{section:Extended GCD Algorithm})
terminates (at line \ref{gcdline:Return}) in the $\mu$-th loop
iteration. We will show in the following that the Partial GCD
Algorithm~I also terminates (at line \ref{gcdline:Return}) in the
$\mu$-th loop iteration.

As shown above, since both the gcd algorithms maintain exactly the
same $s(x)$ and $t(x)$ until $\deg t(x)= \deg \Lambda_f(x)$,
clearly, before the $\mu$-th loop iteration,
\begin{equation}
\deg t(x)<\deg \Lambda_f(x) \leq (N-K)/2
\end{equation}
holds between lines \ref{gcdline:EndWhile}
and~\ref{gcdline:BeginStoppingIf}; moreover, by (\ref{modified
degcheck}) of Lemma \ref{lemma:GCD Loop Invariant.1},
\begin{eqnarray}
\deg \tilde{r}(x)&=& \deg M_n(x)-\deg t(x) \\
                 &>& (N+K)/2 \label{stopcond} \\
                 &>& \deg t(x)+K
\end{eqnarray}
also holds between lines \ref{gcdline:EndWhile}
and~\ref{gcdline:BeginStoppingIf}. Further, from
(\ref{eqn:degt0gt1}), $\deg t(x)> \deg \tilde{t}(x)$ holds as well
between lines \ref{gcdline:EndWhile}
and~\ref{gcdline:BeginStoppingIf}. Therefore,
\begin{equation}
\deg \tilde{r}(x)> (N+K)/2> \deg t(x)+K>\deg \tilde{t}(x)+K
\label{r1t0t1}
\end{equation}
holds between lines \ref{gcdline:EndWhile}
and~\ref{gcdline:BeginStoppingIf} in every but before the $\mu$-th
loop iteration. It then follows after swapping all auxiliary
polynomials in lines \ref{gcdline:Swapr}--\ref{gcdline:Swapt} that
\begin{equation}
\deg r(x)> (N+K)/2 >\deg \tilde{t}(x)+K >\deg t(x)+K
\label{r0t1t0}
\end{equation}
holds between lines~\ref{gcdline:Assignj}
and~\ref{gcdline:BeginWhile} for each subsequent loop iteration.
Then, after executing the \pkw{while} block in the $\mu$-th loop
iteration, the Extended GCD Algorithm in Section
\ref{section:Extended GCD Algorithm} terminates with $r(x)=0$, and
(\ref{eqn:sMtE}) holds; meanwhile, for the Partial GCD
Algorithm~I, we obtain the desired $t(x)$ (with $\deg t(x)= \deg
\Lambda_f(x)$) and $s(x)$, and we have from (\ref{eqn:modified
r0})
\begin{eqnarray}
 r(x)&=& s(x)M_n(x)+t(x)Y(x)\\
     &=& s(x)M_n(x)+t(x)E(x)+t(x)a(x)   \label{gcd.b.1}\\
     &=& t(x)a(x)                       \label{gcd.b.2}
\end{eqnarray}
of $\deg r(x)=\deg t(x)+\deg a(x)< \deg t(x)+K$, where
(\ref{gcd.b.1}) to (\ref{gcd.b.2}) follows from (\ref{eqn:sMtE}).
Finally, since from (\ref{r0t1t0}) $\deg r(x)> \deg t(x)+K$ holds
between lines~\ref{gcdline:Assignj} and~\ref{gcdline:BeginWhile}
but from (\ref{gcd.b.2}) $\deg r(x) < \deg t(x)+K$ holds between
lines \ref{gcdline:EndWhile} and~\ref{gcdline:BeginStoppingIf},
thus the correctness of (\ref{eqn:ModifiedStopping1}) as a
terminating condition is guaranteed; meanwhile from
(\ref{gcd.b.2}) we obtain (\ref{eqn:ModifiedGCDI_rx}). As for
(\ref{eqn:ModifiedStopping2}), since from (\ref{r0t1t0}) $\deg
r(x)> (N+K)/2$ holds between lines~\ref{gcdline:Assignj}
and~\ref{gcdline:BeginWhile} but (from (\ref{gcd.b.2}) and then
(\ref{d_bound.0})) $\deg r(x)<\deg t(x)+K = \deg \Lambda_f(x)+K
\leq (N+K)/2$ holds between lines~\ref{gcdline:EndWhile}
and~\ref{gcdline:BeginStoppingIf}, we thus conclude that
(\ref{eqn:ModifiedStopping2}) can serve as an alternative
terminating condition.

\section*{Appendix C: Proof of Theorem \ref{theorem:AlgorithmII Output}}

In this section, we prove Theorem \ref{theorem:AlgorithmII Output}
in an analogous way as proving Theorem \ref{theorem:AlgorithmI
Output}. The following theorem is an analog of
Lemma~\ref{lemma:GCD Loop Invariant.1}.
\begin{lemma}[GCD Loop Invariant]\label{lemma:GCD Loop Invariant.1.b}
For the Partial GCD Algorithm~II in Section~\ref{section:Modified
GCD Algorithm II}, the condition
\begin{equation}
r(x) = s(x)\cdot M_U(x)+t(x)\cdot E_U(x) \label{eqn:modified r0.b}
\end{equation}
holds both between lines \ref{gcdline:Assignj}
and~\ref{gcdline:BeginWhile} and between lines
\ref{gcdline:EndWhile} and~\ref{gcdline:BeginStoppingIf};
moreover, the conditions
\begin{eqnarray}
\deg r(x) &<& \deg \tilde{r}(x) \label{eqn:degr0r1.b}\\
\deg t(x) &>& \deg \tilde{t}(x) \label{eqn:degt0gt1.b}\\
\deg M_U(x) &=& \deg \tilde{r}(x)+ \deg t(x) \label{modified
degcheck.b}
\end{eqnarray}
hold between lines \ref{gcdline:EndWhile}
and~\ref{gcdline:BeginStoppingIf}.

Specifically, let $\delta_\ell$ denote the degree of $q(x)$
(line~\ref{gcdline:ri}) in the first iteration of the \pkw{while}
block (lines \ref{gcdline:BeginWhile}--\ref{gcdline:EndWhile}) of
the $\ell$-th loop iteration. Then, $\deg t(x)=\deg
\tilde{t}(x)+\delta_\ell=\sum_{v=1}^{\ell} \delta_v$ holds between
lines \ref{gcdline:EndWhile} and~\ref{gcdline:BeginStoppingIf} in
the $\ell$-th loop iteration.
\end{lemma}
The proof of Lemma~\ref{lemma:GCD Loop Invariant.1.b} is the same
as the proof of Lemma~\ref{lemma:GCD Loop Invariant.1}, except for
replacing the $M_n(x)$ in the proof of Lemma~\ref{lemma:GCD Loop
Invariant.1} by $M_U(x)$, and is thus omitted.

We now start to prove Theorem \ref{theorem:AlgorithmII Output}. If
$E(x)=0$, which implies $E_U(x)=0$,
Theorem~\ref{theorem:AlgorithmII Output} holds obviously; we thus
prove in the following only the case where $E(x)\neq 0$. For the
Partial GCD Algorithm~II of Section \ref{section:Modified GCD
Algorithm II}, let $g$ denote the largest integer such that $x^g$
of either $r(x)$ or of $\tilde{r}(x)$ is unknown. Clearly, with
$M_U(x)$ and $E_U(x)$ as inputs, the algorithm starts with $g =
-1$. Moreover, let $h\eqdef \max\{\deg r(x),\deg \tilde{r}(x)\}$.
Clearly, the algorithm starts with $h=\deg M_U(x)=N-K$.

\begin{lemma}\label{lemma:GCD Algorithm.1.b}
For the Partial GCD Algorithm~II in Section \ref{section:Modified
GCD Algorithm II}, let $\delta_\ell$ denote the degree of $q(x)$
in the first iteration of the \pkw{while} block (lines
\ref{gcdline:BeginWhile}--\ref{gcdline:EndWhile}) of the $\ell$-th
loop iteration.  If $h-g>2\delta_\ell$ holds between
lines~\ref{gcdline:Assignj} and~\ref{gcdline:BeginWhile}, then the
value of $q(x)$ (line~\ref{gcdline:ri}) throughout the \pkw{while}
block in the $\ell$-th loop iteration is exactly the same as the
corresponding one of the Extended GCD Algorithm of
Section~\ref{section:Extended GCD Algorithm} in the same loop
iteration. In addition, $g=-1+\sum_{v=1}^{\ell} \delta_v$ and
$h=N-K-\sum_{v=1}^{\ell} \delta_v$ both hold between lines
\ref{gcdline:EndWhile} and~\ref{gcdline:BeginStoppingIf} in the
$\ell$-th loop iteration.
\end{lemma}
The proof is similar to that of Lemma~\ref{lemma:GCD Algorithm.1}
and is thus omitted. Since $h-g=N-K+1$ holds between
lines~\ref{gcdline:Assignj} and~\ref{gcdline:BeginWhile} in the
first loop iteration, it follows from Lemma~\ref{lemma:GCD
Algorithm.1.b} that if $2 \sum_{v=1}^{\ell} \delta_v <  N-K+1$,
then, from the first to the $\ell$-th loop iteration, $q(x)$ and
thus $s(x)$ and $t(x)$ are exactly the same as in the Extended GCD
Algorithm.  Moreover, from Lemma~\ref{lemma:GCD Loop
Invariant.1.b}, $\deg t(x)=\sum_{v=1}^{\ell} \delta_v$ holds
between lines \ref{gcdline:EndWhile}
and~\ref{gcdline:BeginStoppingIf}. In order to obtain
(\ref{eqn:gcdtLambdaf}), which implies that $\deg t(x)=\deg
\Lambda_f(x)$, it turns out that if
\begin{equation}
2\deg \Lambda_f(x)\leq N-K,
\end{equation}
which agrees with (\ref{d_bound.0}), then the algorithm maintains
exactly the same $s(x)$ and $t(x)$ as the Extended GCD Algorithm
of Section~\ref{section:Extended GCD Algorithm} until $\deg
t(x)=\deg \Lambda_f(x)$.

It remains to argue the validity of
(\ref{eqn:ModifiedStrippedStopping1}) and
(\ref{eqn:ModifiedStrippedStopping2}) as appropriate terminating
conditions. Assume that (\ref{d_bound.0}) is satisfied and suppose
the Extended GCD Algorithm (in Section~\ref{section:Extended GCD
Algorithm}) terminates (at line \ref{gcdline:Return}) in the
$\mu$-th loop iteration. As shown above, it has been clear that
the Extended GCD Algorithm in Section~\ref{section:Extended GCD
Algorithm} and the Partial GCD Algorithm~II maintain exactly the
same $s(x)$ and $t(x)$ until $\deg t(x)= \deg \Lambda_f(x)$. Thus,
before the $\mu$-th loop iteration
\begin{equation}
\deg t(x)<\deg \Lambda_f(x) \leq (N-K)/2
\end{equation}
holds between lines \ref{gcdline:EndWhile}
and~\ref{gcdline:BeginStoppingIf}; moreover, by (\ref{modified
degcheck.b}) of Lemma~\ref{lemma:GCD Loop Invariant.1.b},
\begin{eqnarray}
\deg \tilde{r}(x)&=& \deg M_U(x)-\deg t(x) \\
                 &>& (N-K)/2 \label{stopcond.b} \\
                 &>& \deg t(x)
\end{eqnarray}
also holds between lines \ref{gcdline:EndWhile}
and~\ref{gcdline:BeginStoppingIf} for the Partial GCD
Algorithm~II. Further, from (\ref{eqn:degt0gt1.b}), $\deg t(x)>
\deg \tilde{t}(x)$ holds as well between lines
\ref{gcdline:EndWhile} and~\ref{gcdline:BeginStoppingIf}.
Therefore, for the Partial GCD Algorithm~II,
\begin{equation}
\deg \tilde{r}(x)> (N-K)/2> \deg t(x)>\deg \tilde{t}(x)
\label{r1t0t1.b}
\end{equation}
holds between lines \ref{gcdline:EndWhile}
and~\ref{gcdline:BeginStoppingIf} in every but before the $\mu$-th
loop iteration. It then follows after swapping all auxiliary
polynomials in lines \ref{gcdline:Swapr}--\ref{gcdline:Swapt} that
\begin{equation}
\deg r(x)> (N-K)/2 >\deg \tilde{t}(x) >\deg t(x) \label{r0t1t0.b}
\end{equation}
holds between lines~\ref{gcdline:Assignj}
and~\ref{gcdline:BeginWhile} for each subsequent loop iteration.
Then, after executing the \pkw{while} block in the $\mu$-th loop
iteration, we obtain the desired $t(x)$ (with $\deg t(x)= \deg
\Lambda_f(x)$) and $s(x)$ that coincide with the corresponding
ones of the Extended GCD Algorithm in Section
\ref{section:Extended GCD Algorithm}; thus $t(x)$ and $s(x)$
(in the Partial GCD Algorithm~II) at this moment satisfy both
(\ref{eqn:modified r0.b}) and (\ref{eqn:sMtE}). From
(\ref{eqn:sMtE}), we have
\begin{equation}
 -s(x)M_n(x)=t(x)E(x) \label{gcd.0}
\end{equation}
with $\deg s(x)< \deg t(x)$. Note that (\ref{gcd.0}) can also be
written as
\begin{equation}
 -s(x)(x^K M_U(x)+M_L(x))=t(x)(x^K E_U(x)+E_L(x)), \label{gcd.1}
\end{equation}
where $M_U(x)$ and $E_U(x)$ are defined in Section
\ref{section:Modified GCD Algorithm II} and $M_L(x)= M_n(x)-x^K
M_U(x)$ and $E_L(x)= E(x)-x^K E_U(x)$. Further, let $V(x)\eqdef
-s(x)M_L(x)-t(x)E_L(x)=\sum_{\ell=0}V_{\ell} \, x^\ell$, which is
of degree $\deg V(x)\leq (K-1)+\deg t(x)$ because $\deg s(x)< \deg
t(x)$. Equation (\ref{gcd.1}) can then be written as
\begin{equation}
x^K \left(s(x)M_U(x)+t(x)E_U(x)\right)= V(x). \label{gcd.2}
\end{equation}
Observing the left hand side of (\ref{gcd.2}), we know that all
the terms on the right hand side of (\ref{gcd.2}) of degree less
than $K$ will vanish. Thus, we have the following equivalent
expression for (\ref{gcd.2}):
\begin{equation}
s(x)M_U(x)+t(x)E_U(x) = V_U(x) \label{gcd.3}
\end{equation}
where $V_U(x) \eqdef \sum_{\ell=0}V_{K+\ell} \, x^\ell$ has degree
\begin{eqnarray}
 \deg V_U(x)&=&\deg V(x)-K \nonumber \\
            &\leq&(K-1)+\deg t(x)-K \nonumber \\
            &<& \deg t(x).  \label{stop.0}
\end{eqnarray}
Comparing (\ref{gcd.3}) with (\ref{eqn:modified r0.b}) and from
(\ref{stop.0}), clearly, $\deg r(x)=\deg V_U(x)<\deg t(x)$, which
coincides with (\ref{eqn:ModifiedStrippedStopping1}), holds
between lines \ref{gcdline:EndWhile}
and~\ref{gcdline:BeginStoppingIf} in the $\mu$-th loop iteration.
Thus, the correctness of (\ref{eqn:ModifiedStrippedStopping1}) as
a terminating condition is guaranteed (because from
(\ref{r0t1t0.b}) $\deg r(x)>\deg t(x)$ holds between
lines~\ref{gcdline:Assignj} and~\ref{gcdline:BeginWhile}). On the
other hand, since from (\ref{r0t1t0.b}) $\deg r(x)> (N-K)/2$ holds
between lines~\ref{gcdline:Assignj} and~\ref{gcdline:BeginWhile}
but $\deg r(x)<\deg t(x)=\deg \Lambda_f(x)\leq (N-K)/2$ holds
between lines \ref{gcdline:EndWhile}
and~\ref{gcdline:BeginStoppingIf}, we thus conclude that
(\ref{eqn:ModifiedStrippedStopping2}) can serve as an alternative
terminating condition.

\newcommand{\COM}{IEEE Trans.\ Communications}
\newcommand{\COMMag}{IEEE Communications Mag.}
\newcommand{\IT}{IEEE Trans.\ Information Theory}
\newcommand{\JSAC}{IEEE J.\ Select.\ Areas in Communications}
\newcommand{\SP}{IEEE Trans.\ Signal Proc.}
\newcommand{\SPMag}{IEEE Signal Proc.\ Mag.}
\newcommand{\ProcIEEE}{Proceedings of the IEEE}

\end{document}